\documentclass{article} 
\usepackage{nips14submit_e,times}
\usepackage{hyperref}
\usepackage{url}
\usepackage{caption}
\usepackage{naesseth}
\usepackage{abbreviations}
\usepackage{amsthm}
\usepackage[numbers]{natbib}
\renewcommand{\bibfont}{\small}
\usepackage{tikz}
\usepackage{graphicx}
\usepackage{subcaption}
\usepackage{wrapfig}
\usepackage{algorithm}
\usepackage{algorithmic}

\usepackage{color}
\usepackage[textwidth=2cm, textsize=scriptsize]{todonotes}

\newtheorem{lem}{Lemma}

\newcommand\filt{\mathcal{F}}
\newcommand\tcdots{\!\!\cdot\!\cdot\!\cdot\!\!}

\title{Sequential Monte Carlo for Graphical Models}

\author{
Christian A. Naesseth \\
Div. of Automatic Control\\
Link\"oping University\\
Link\"oping, Sweden \\
\texttt{chran60@isy.liu.se} \\
\And
Fredrik Lindsten \\
Dept. of Engineering\\
The University of Cambridge \\
Cambridge, UK \\
\texttt{fsml2@cam.ac.uk} \\
\And
Thomas B. Sch\"on \\
Dept. of Information Technology\\
Uppsala University\\
Uppsala, Sweden \\
\texttt{thomas.schon@it.uu.se} \\
}

\nipsfinalcopy 

\begin{document}

\maketitle

\begin{abstract}
We propose a new framework for how to use sequential Monte Carlo (\smc) algorithms for inference in probabilistic graphical models (\pgm). Via a sequential decomposition of the \pgm we find a sequence of auxiliary distributions defined on a monotonically increasing sequence of probability spaces. By targeting these auxiliary distributions using \smc we are able to approximate the full joint distribution defined by the \pgm. One of the key merits of the
 \smc sampler is that it provides an unbiased estimate of the partition function of the model.
We also show how it can be used within a particle Markov chain Monte Carlo framework in order to construct
high-dimensional block-sampling algorithms for general \pgm{s}.
\end{abstract}

\section{Introduction}
\label{sec:introduction}
%
Bayesian inference in statistical models involving a large number of latent random variables is in general a difficult problem. This renders inference methods that are capable of efficiently utilizing structure important tools. Probabilistic Graphical Models (\pgm{s}) are an intuitive and useful way to represent and make use of underlying structure in probability distributions with many interesting areas of applications~\citep{jordan2004graphical}. 

%
Our main contribution is a new framework for constructing non-standard (auxiliary) target distributions of \pgm{s}, utilizing what we call a \emph{sequential decomposition} of the underlying factor graph, to be targeted by a sequential Monte Carlo (\smc) sampler. This construction enables us to make use of \smc methods developed and studied over the last $20$ years, to approximate the full joint distribution defined by the \pgm. 
As a byproduct, the \smc algorithm provides an unbiased estimate of the partition function (normalization constant). We show how the proposed method can be used as an alternative to standard methods such as the Annealed Importance Sampling (\ais) proposed in \citep{neal2001annealed}, when estimating the partition function.
We also make use of the proposed \smc algorithm to design efficient, high-dimensional \mcmc kernels for the latent variables of the \pgm in a particle \mcmc framework. This enables inference about the latent variables as well as learning of unknown model parameters in an \mcmc setting.

%

During the last decade there has been substantial work on how to leverage \smc algorithms \cite{doucet2001sequential} to solve inference problems in \pgm{s}. The first approaches were PAMPAS \citep{isard2003pampas} and nonparametric belief propagation by \citet{sudderth2003nonparametric, sudderth2010nonparametric}. Since then, several different variants and refinements have been proposed by \eg \citet{briers2005sequential, ihler2009particle, frank2009particle}. They all rely on various particle approximations of messages sent in a loopy belief propagation algorithm. This means that in general, even in the limit of Monte Carlo samples, they are approximate methods. Compared to these approaches our proposed methods are consistent and provide an unbiased estimate of the normalization constant as a by-product.

Another branch of \smc-based methods for graphical models has been suggested by \citet{hamze2005hot}. Their method builds on the \smc sampler by \citet{del2006sequential}, where the initial target is a spanning tree of the original graph and subsequent steps add edges according to an annealing schedule. \citet{everitt2012bayesian} extends these ideas to learn parameters using particle \mcmc \cite{andrieu2010particle}. Yet another take is provided by \citet{de2007conditional}, where an \smc sampler is combined with mean field approximations. Compared to these methods we can handle both non-Gaussian and/or non-discrete interactions between variables and there is no requirement to perform \mcmc steps within each \smc step.

The left-right methods described by \citet{wallach2009evaluation} and extended by \citet{buntine2009estimating} to estimate the likelihood of held-out documents in topic models are somewhat related in that they are \smc-inspired. However, these are not actual \smc algorithms and they do not produce an unbiased estimate of the partition function for finite sample set. On the other hand, a particle learning based approach was recently proposed by \citet{scott2013arecursive} and it can be viewed as a special case of our method for this specific type of model.


\section{Graphical models} 
\label{sec:graphicalmodels}
A graphical model is a probabilistic model which {\em factorizes} according
to the structure of an underlying graph $\G = \{\V, \E\}$, with vertex set $\V$ and edge set
$\E$. By this we mean that the joint probability density function (\pdf) of the set of random variables indexed by $\V$,
$\xV \eqdef \crange{x_1}{x_{|\V|}}$, can be represented as a product of factors over the cliques of the
graph:
\begin{equation}
p( \xV ) = \frac{1}{Z} \prod_{C \in \C} \psi_C (\X_C),
\label{eq:pfactors}
\end{equation}
where $\C$ is the set of cliques in $\G$, $\psi_C$ is the factor for clique $C$,
and $Z = \int \prod_{C \in \C} \psi_C (x_C) \rmd\xV$ is the partition function. 

\begin{wrapfigure}{r}{0.35\textwidth}
\vspace{-10pt}
\centering
\begin{subfigure}[b]{0.35\textwidth}
\centering
\begin{tikzpicture}[node distance=1.06cm,main node/.style={circle,draw, scale=0.7}]
  \node[main node] (1) {$x_1$};
  \node[main node] (2) [right of=1] {$x_2$};
  \node[main node] (3) [above right of=2] {$x_3$};
  \node[main node] (4) [below right of=2] {$x_4$};
  \node[main node] (5) [below right of=3] {$x_5$};

  \path (1) edge (2);
  \path (2) edge (3);
  \path (2) edge (4);
  \path (3) edge (4);
  \path (3) edge (5);
  \path (4) edge (5);
\end{tikzpicture}
\caption{Undirected graph.}
\label{fig:graph}
\end{subfigure}\vspace{5pt}
\begin{subfigure}[b]{0.35\textwidth}
\centering
\begin{tikzpicture}[node distance=1cm,main node/.style={circle,draw, scale=0.7},factor node/.style={rectangle,draw, scale=0.7}]
  \node[main node] (1) {$x_1$};
  \node[factor node] (2) [right of=1] {$\psi_1$};
  \node[main node] (3) [right of=2] {$x_2$};
  \node[factor node] (4) [right of=3] {$\psi_2$};
  \node[main node] (5) [above right of=4] {$x_3$};
  \node[main node] (6) [below right of=4] {$x_4$};
  \node[factor node] (7) [right of=5] {$\psi_3$};
  \node[factor node] (8) [right of=6] {$\psi_4$};
  \node[main node] (9) [below right of=7] {$x_5$};
  \node[factor node] (10) [right of=9] {$\psi_5$};

  \path (1) edge (2);
  \path (2) edge (3);
  \path (3) edge (4);
  \path (4) edge (5);
  \path (4) edge (6);
  \path (5) edge (7);
  \path (7) edge (9);
  \path (6) edge (8);
  \path (8) edge (9);
  \path (9) edge (10);
\end{tikzpicture}
\caption{Factor graph.}
\label{fig:factorgraph}
\end{subfigure}
\caption{Undirected \pgm and a corresponding
  factor graph.}
\vspace{-30pt}
\label{fig:PGM}
\vspace{-25pt}
\end{wrapfigure}
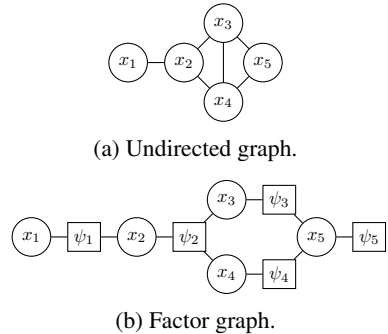

We will frequently use the notation $\X_I = \bigcup_{i \in I}\{x_i\}$ for some subset $I \subseteq
\set{1}{|\V|}$ and we write $\setX_I$ for the range of $\X_I$ (\ie, $\X_I \in \setX_I)$.
To make the interactions between the random variables explicit we define a {\em factor
graph} $\F = \{\V, \Psi, \E' \}$ corresponding to $\G$. The factor graph consists of two types of
vertices, the original set of random variables $\xV$ and the factors $\Psi = \{\psi_C : C \in
\C\}$. The edge set $\E'$ consists only of edges from variables to factors. In
Figure~\ref{fig:graph} we show a simple toy example of an undirected graphical model, and one possible corresponding factor graph, Figure~\ref{fig:factorgraph}, making the dependencies explicit. Both directed and undirected graphs can be represented by factor graphs.


\section{Sequential Monte Carlo}\label{sec:smc}
In this section we propose a way to sequentially decompose a graphical model which we then make use of
to design an \smc algorithm for the \pgm.

\subsection{Sequential decomposition of graphical models}\label{sec:seqdecomp}
\smc methods can be used to approximate a sequence of probability distributions on a sequence of
probability spaces of increasing dimension.  This is done by recursively updating a set of 
samples---or \emph{particles}---with corresponding nonnegative importance weights.
The typical scenario is that of state inference in state-space models, where the probability
distributions targeted by the \smc sampler are the joint smoothing distributions of a sequence of
latent states conditionally on a sequence of observations; see \eg, \citet{DoucetJ:2011} for
applications of this type.  However, \smc is not limited to these cases and it is applicable to a much
wider class of models.

To be able to use \smc for inference in \pgm{s} we have to define a sequence of 
target distributions. However, these target distributions
\emph{do not} have to be marginal distributions under $p(\xV)$.
Indeed, as long as the sequence of target distributions is constructed in such a way
that, at some final iteration, we recover $p(\xV)$, all the intermediate target distributions
may be chosen quite arbitrarily.

This is key to our development, since it lets us use the structure of the \pgm
to define a sequence of intermediate target distributions for the sampler.
We do this by a so called {\em sequential decomposition} of the graphical model.
This amounts to simply adding factors to the target distribution, from the product of factors in~\eqref{eq:pfactors}, at each step of the algorithm and iterate until all
the factors have been added.
Constructing an artificial sequence of intermediate target distributions for an \smc sampler is a
simple, albeit underutilized, idea as it opens up for using \smc samplers
for inference in a wide range of probabilistic models; see \eg, \citet{Bouchard-CoteSJ:2012,del2006sequential} for a few applications of this approach.

Given a graph $\G$ with cliques $\C$, let $\{\psi_k\}_{k=1}^K$ be a sequence of factors defined as follows
$\psi_k(\xk) = \prod_{C \in \C_k} \psi_C (\X_C)$, 
where $\C_k \subset \C$ are chosen such that $\bigcup_{k=1}^K \C_k = \C$ and $\C_i \cap \C_j = \emptyset,~i \neq j$,
and where $\ind_k \subseteq \set{1}{|\V|}$ is the index set of the variables in the domain of $\psi_k$,
  $\ind_k = \bigcup_{C\in \C_k} C$. 
We emphasize that the cliques in $\C$ need not be maximal. In fact even auxiliary factors may be introduced to allow for \eg annealing between distributions.
It follows that the \pdf in \eqref{eq:pfactors} can be written as
  $p(\xV) = \frac{1}{Z}\prod_{k = 1}^K \psi_k(\xk)$. 
Principally, the choices and the ordering of the $\C_k$'s is arbitrary, but in practice it
will affect the performance of the proposed sampler. However, in many common \pgm{s} an intuitive ordering can be deduced
from the structure of the model, see Section~\ref{sec:experiments}.

The sequential decomposition of the \pgm is then based on the auxiliary quantities
 $ \tgammak (\Xk) \eqdef \prod_{\ell = 1}^k \psi_\ell (\X_{\ind_\ell})$,
with $\Ind_k \eqdef \bigcup_{\ell=1}^k \ind_\ell$,  for $k \in \set{1}{K}$.
By construction, $\Ind_K =\V$ and the joint \pdf $p(\Xk[K])$ will be proportional to $\tgamma_K(\Xk[K])$.
Consequently, by using 
$  \tgammak (\Xk)$
as the basis for the target sequence for an \smc sampler,
we will obtain the correct target distribution at iteration $K$.
%
However, a further requirement for this to be possible is that all the functions in the sequence
are normalizable. For many graphical models this is indeed the case, and then we can use $\tgammak (\Xk) $, $k = 1$ to $K$, 
directly as our sequence of intermediate target densities. If, however, $\int \tgammak(\Xk)\rmd \Xk = \infty$ for some $k < K$,
an easy remedy is to modify the target density to ensure normalizability.
This is done by setting $\gammak(\Xk) = \tgammak(\Xk) q_k(\Xk)$, where $q_k(\Xk)$ is choosen so that
$  \int \gammak(\Xk) \rmd \Xk < \infty$.
We set $q_K(\Xk[K]) \equiv 1$ to make sure that $\gamma_K( \Xk[K] ) \propto p(\Xk[k])$.
Note that the integral $\int \gammak(\Xk) \rmd \Xk$ 
need not be computed explicitly,
as long as it can be established that it is finite. With this modification we obtain
a
sequence of unnormalized intermediate target densities for the \smc sampler as $\gamma_1(\Xk[1]) =
q_1(\Xk[1]) \psi_1 (\Xk[1])$ and $\gammak(\Xk) = \gamma_{k-1}(\Xk[k-1]) \frac{q_k(\Xk)}{q_{k-1}(\Xk[k-1])} \psi_k (\xk)$
for $k = \range{2}{K}$. The corresponding normalized \pdf{s} are given by $\bar{\gamma}_k(\Xk) =
\gammak(\Xk)/Z_k$, where $Z_k = \int \gammak(\Xk) \rmd \Xk$.
Figure~\ref{fig:decomp} shows two examples of possible subgraphs when applying
the decomposition, in two different ways, to the factor graph example in Figure~\ref{fig:PGM}.

\newcommand{\sca}{0.6}
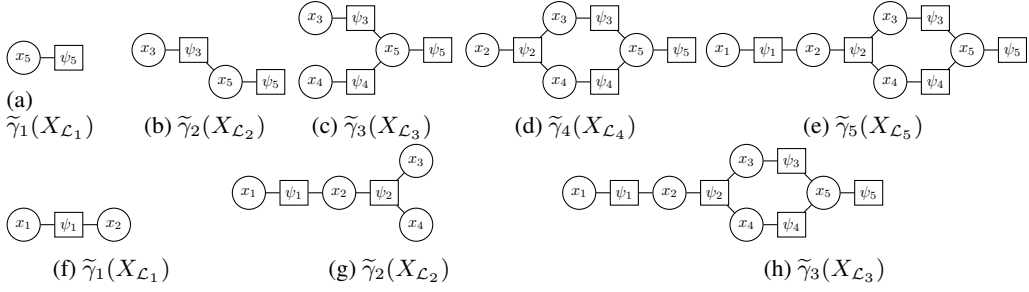
\begin{figure}[t]
\begin{center}
\begin{subfigure}[b]{1.0\textwidth}
\begin{subfigure}[b]{0.1\textwidth}
\begin{tikzpicture}[node distance=1cm,main node/.style={circle,draw,scale=\sca},factor node/.style={rectangle,draw,scale=\sca}]
  \node[main node] (9) {$x_5$};
  \node[factor node] (10) [right of=9] {$\psi_5$};
  \path (9) edge (10);
\end{tikzpicture}
\caption{$\tgamma_1 (\Xk[1])$}
\end{subfigure}
~
\begin{subfigure}[b]{0.14\textwidth}
\begin{tikzpicture}[node distance=1cm,main node/.style={circle,draw,scale=\sca},factor node/.style={rectangle,draw,scale=\sca}]
  \node[main node] (5)  {$x_3$};
  \node[factor node] (7) [right of=5] {$\psi_3$};
  \node[main node] (9) [below right of=7] {$x_5$};
  \node[factor node] (10) [right of=9] {$\psi_5$};
  \path (5) edge (7);
  \path (7) edge (9);
  \path (9) edge (10);
\end{tikzpicture}
\caption{$\tgamma_2 (\Xk[2])$}
\end{subfigure}
~
\begin{subfigure}[b]{0.14\textwidth}
\begin{tikzpicture}[node distance=1cm,main node/.style={circle,draw,scale=\sca},factor node/.style={rectangle,draw,scale=\sca}]
  \node[main node] (5) {$x_3$};
  \node[factor node] (7) [right of=5] {$\psi_3$};
  \node[factor node] (8) [below left of=9] {$\psi_4$};
  \node[main node] (9) [below right of=7] {$x_5$};
  \node[main node] (6) [left of=8] {$x_4$};
  \node[factor node] (10) [right of=9] {$\psi_5$};
  \path (5) edge (7);
  \path (7) edge (9);
  \path (6) edge (8);
  \path (8) edge (9);
  \path (9) edge (10);
\end{tikzpicture}
\caption{$\tgamma_3 (\Xk[3])$}
\end{subfigure}
~
\begin{subfigure}[b]{0.21\textwidth}
\begin{tikzpicture}[node distance=1cm,main node/.style={circle,draw,scale=\sca},factor node/.style={rectangle,draw,scale=\sca}]
  \node[factor node] (8) {$\psi_4$};
  \node[main node] (6) [left of=8] {$x_4$};
  \node[factor node] (4) [above left of=6] {$\psi_2$};
  \node[main node] (3) [left of=4] {$x_2$};
  \node[main node] (5) [above right of=4] {$x_3$};
  \node[factor node] (7) [right of=5] {$\psi_3$};
  \node[main node] (9) [below right of=7] {$x_5$};
  \node[factor node] (10) [right of=9] {$\psi_5$};
  \path (3) edge (4);
  \path (4) edge (5);
  \path (4) edge (6);
  \path (5) edge (7);
  \path (7) edge (9);
  \path (6) edge (8);
  \path (8) edge (9);
  \path (9) edge (10);
\end{tikzpicture}
\caption{$\tgamma_4 (\Xk[4])$}
\end{subfigure}
~
\begin{subfigure}[b]{0.3\textwidth}
\begin{tikzpicture}[node distance=1cm,main node/.style={circle,draw,scale=\sca},factor node/.style={rectangle,draw,scale=\sca}]
  \node[main node] (1) {$x_1$};
  \node[factor node] (2) [right of=1] {$\psi_1$};
  \node[main node] (3) [right of=2] {$x_2$};
  \node[factor node] (4) [right of=3] {$\psi_2$};
  \node[main node] (5) [above right of=4] {$x_3$};
  \node[main node] (6) [below right of=4] {$x_4$};
  \node[factor node] (7) [right of=5] {$\psi_3$};
  \node[factor node] (8) [right of=6] {$\psi_4$};
  \node[main node] (9) [below right of=7] {$x_5$};
  \node[factor node] (10) [right of=9] {$\psi_5$};

  \path (1) edge (2);
  \path (2) edge (3);
  \path (3) edge (4);
  \path (4) edge (5);
  \path (4) edge (6);
  \path (5) edge (7);
  \path (7) edge (9);
  \path (6) edge (8);
  \path (8) edge (9);
  \path (9) edge (10);
\end{tikzpicture}
\caption{$\tgamma_5 (\Xk[5])$}
\end{subfigure}
\end{subfigure}

\begin{subfigure}[b]{1.0\textwidth}
\begin{subfigure}[b]{0.2\textwidth}
\begin{tikzpicture}[node distance=1cm,main node/.style={circle,draw,scale=\sca},factor node/.style={rectangle,draw,scale=\sca}]
  \node[main node] (1) {$x_1$};
  \node[factor node] (2) [right of=1] {$\psi_1$};
  \node[main node] (3) [right of=2] {$x_2$};
 
  \path (1) edge (2);
  \path (2) edge (3);
\end{tikzpicture}
\caption{$\tgamma_1 (\Xk[1])$}
\end{subfigure}
~
\begin{subfigure}[b]{0.3\textwidth}
\begin{tikzpicture}[node distance=1cm,main node/.style={circle,draw,scale=\sca},factor node/.style={rectangle,draw,scale=\sca}]
  \node[main node] (1) {$x_1$};
  \node[factor node] (2) [right of=1] {$\psi_1$};
  \node[main node] (3) [right of=2] {$x_2$};
  \node[factor node] (4) [right of=3] {$\psi_2$};
  \node[main node] (5) [above right of=4] {$x_3$};
  \node[main node] (6) [below right of=4] {$x_4$};

  \path (1) edge (2);
  \path (2) edge (3);
  \path (3) edge (4);
  \path (4) edge (5);
  \path (4) edge (6);
\end{tikzpicture}
\caption{$\tgamma_2 (\Xk[2])$}
\end{subfigure}
~
\begin{subfigure}[b]{0.5\textwidth}
\begin{tikzpicture}[node distance=1cm,main node/.style={circle,draw,scale=\sca},factor node/.style={rectangle,draw,scale=\sca}]
  \node[main node] (1) {$x_1$};
  \node[factor node] (2) [right of=1] {$\psi_1$};
  \node[main node] (3) [right of=2] {$x_2$};
  \node[factor node] (4) [right of=3] {$\psi_2$};
  \node[main node] (5) [above right of=4] {$x_3$};
  \node[main node] (6) [below right of=4] {$x_4$};
  \node[factor node] (7) [right of=5] {$\psi_3$};
  \node[factor node] (8) [right of=6] {$\psi_4$};
  \node[main node] (9) [below right of=7] {$x_5$};
  \node[factor node] (10) [right of=9] {$\psi_5$};

  \path (1) edge (2);
  \path (2) edge (3);
  \path (3) edge (4);
  \path (4) edge (5);
  \path (4) edge (6);
  \path (5) edge (7);
  \path (7) edge (9);
  \path (6) edge (8);
  \path (8) edge (9);
  \path (9) edge (10);
\end{tikzpicture}
\caption{$\tgamma_3 (\Xk[3])$}
\end{subfigure}
\end{subfigure}
\caption{Examples of five- (top) and three-step (bottom) sequential decomposition of Figure \ref{fig:PGM}.}
\label{fig:decomp}
\end{center}
\vspace{-1.1\baselineskip}
\end{figure}

\subsection{Sequential Monte Carlo for \pgm{s}}
At iteration $k$, the \smc sampler approximates the target distribution $\bar{\gamma}_k$ by
a collection of weighted particles $\{\Xk^i, w_k^i\}_{i=1}^N$. These samples define an empirical
point-mass approximation of the target distribution. In what follows, we shall use the notation
 $ \xi_k \eqdef \X_{\ind_k \setminus \Ind_{k-1}}$
to refer to the collection of random variables that are in the domain of $\gamma_k$, but not in the domain of $\gamma_{k-1}$.
This corresponds to the collection of random variables, with which the particles are augmented at each iteration.

Initially, $\bar{\gamma}_1$ is approximated by importance sampling. We proceed inductively and assume that we have at hand a weighted sample
$\{\Xk[k-1]^i, w_{k-1}^i\}_{i=1}^N$, approximating $\bar{\gamma}_{k-1}(X_{\Ind_{k-1}})$.
This sample is propagated forward by simulating, conditionally independently
given the particle generation up to iteration $k-1$, and drawing an {\em ancestor index}
 $a_k^i$ with $\Prb(a_k^i = j) \propto \nu_{k-1}^j w_{k-1}^j$, $j = \range{1}{\Np}$,
where $\nu_{k-1}^i \eqdef \nu_{k-1}(\Xk[k-1]^{i})$---known as adjustment multiplier weights---are used in
the auxiliary SMC framework to adapt the resampling procedure to the current target density
$\bar{\gamma}_k$ \citep{PittS:1999}. 
Given the ancestor indices, we simulate particle increments $\{ \xi_k^i\}_{i=1}^N$ from a proposal density $\xi_k^i \sim r_k (\cdot | \Xk[k-1]^{a_k^i} )$ on $\setX_{\ind_k \setminus \Ind_{k-1}}$,
and augment the particles as $\Xk^i \eqdef \Xk[k-1]^{a_k^i} \cup \xi_k^i$.

After having performed this procedure for the $N$ ancestor indices and
particles, they are assigned importance weights $w_k^i = W_k (\Xk^i)$.
The weight function, for $k \geq 2$, is given by
\begin{align}
  \label{eq:smc:wf}
  W_k (\Xk) &= \frac{\displaystyle \gammak (\Xk)}{  \gamma_{k-1}(\Xk[k-1]) \nu_{k-1}(\Xk[k-1])  r_k (\xi_k | \Xk[k-1] )},
\end{align}
where, again, we write $\xi_k = \X_{\ind_k \setminus \Ind_{k-1}}$.
We give a summary of the SMC method in Algorithm~\ref{alg:smc}.

\begin{wrapfigure}[12]{r}{0.6\textwidth}
  \vspace*{-7.5mm}
  \begin{minipage}{0.6\textwidth}
    \begin{algorithm}[H]
      \caption{Sequential Monte Carlo (SMC)}
      \label{alg:smc}
      \begin{algorithmic}
        \STATE {\em Perform each step for $i = 1,\ldots,N$.}
        \STATE Sample $\Xk[1]^i \sim r_1(\cdot)$.
        \STATE Set $w_1^i = \gamma_1(\Xk[1]^i)/r_1(\Xk[1]^i)$.
        \FOR{$k=2$ {\bfseries to} $K$}
        \STATE Sample $a_k^{i}$ according to~$\Prb(a_k^i = j) = \frac{ \nu_{k-1}^j w_{k-1}^j }{ \sum_{l} \nu_{k-1}^l w_{k-1}^l }$.
        \STATE Sample $\xi_k^i \sim r_k (\cdot | \Xk[k-1]^{a_k^i} )$ and set $\Xk^i = \Xk[k-1]^{a_k^i} \cup \xi_k^i$.
        \STATE Set $w_k^i = W_k (\Xk^i)$.
        \ENDFOR
      \end{algorithmic}
    \end{algorithm}
  \end{minipage}
\end{wrapfigure}

In the case that $\ind_k \setminus \Ind_{k-1} = \emptyset$ for some $k$, resampling and propagation steps are
superfluous. The easiest way to handle this is to simply skip these steps and directly compute importance weights. 
An alternative approach is to bridge the two target distributions $\bar\gamma_{k-1}$ and $\bar\gamma_{k}$ 
similarly to \citet{del2006sequential}.

Since the proposed sampler for \pgm{s} falls within a general \smc framework, standard
convergence analysis applies. See \eg, 
\citet{DelMoral:2004} for a comprehensive collection
of theoretical results on consistency, central limit theorems, and non-asymptotic bounds for \smc samplers.

The choices of proposal density and adjustment multipliers can quite significantly affect the
performance of the sampler. 

It follows from \eqref{eq:smc:wf} that $W_k(\Xk) \equiv 1$ if we choose
$\nu_{k-1}(\Xk[k-1]) = \int \frac{\gammak (\Xk) }{  \gamma_{k-1}(\Xk[k-1]) } \rmd \xi_{k}$
and $r_k(\xi_k | \Xk[k-1]) = \frac{\gammak (\Xk) }{ \nu_{k-1}(\Xk[k-1]) \gamma_{k-1}(\Xk[k-1]) }$.
In this case, the \smc sampler is said to be \emph{fully adapted}.

\subsection{Estimating the partition function}\label{sec:smc:Zhat}
The partition function of a graphical model is a very
interesting quantity in many applications. Examples include
likelihood-based learning of the parameters of the \pgm, statistical
mechanics where it is related to the free energy of a system of objects, and
information theory where it is related to the capacity of a
channel.
However, as stated by \citet{hamze2005hot}, estimating the partition function of a loopy graphical model is
a ``notoriously difficult'' task. 
Indeed, even for discrete problems simple and accurate estimators have proved to be elusive,
and \mcmc methods do not provide any simple way of computing the partition function.

On the contrary, \smc provides a straightforward estimator of the normalizing constant (\ie the partition function),
given as a byproduct of the sampler according to,
\begin{equation}
  \label{eq:smc:Zhat}
  \widehat Z_k^N \eqdef \left( \frac{1}{N} \sum_{i=1}^{N} w_k^i \right) \left\{  \prod_{\ell = 1}^{k-1} \frac{1}{N} \sum_{i=1}^N \nu_\ell^i w_\ell^i  \right\}.
\end{equation}
It may not be obvious to see why \eqref{eq:smc:Zhat} is a natural estimator of the normalizing constant $Z_k$.
However, a by now well known result is that this \smc-based estimator is unbiased. This result is due to \citet[Proposition~7.4.1]{DelMoral:2004}
and, for the special case of inference in state-space models, it has also been established by
\citet{PittSGK:2012}. For completeness we also offer a proof using the present notation in the appendix.

Since $Z_K = Z$, we thus obtain an estimator of the partition function of the \pgm at iteration $K$
of the sampler. Besides from being unbiased, this estimator is also consistent and asymptotically
normal; see \citet{DelMoral:2004}.

In \citep{NaessethLS:2014IT} we
have studied a specific information theoretic application (computing the capacity of a two-dimensional channel) and
inspired by the algorithm proposed here we were able to design a sampler with significantly improved
performance compared to the previous state-of-the-art.


\newcommand\Ipgas{\Ind_K}

\section{Particle MCMC and partial blocking}\label{sec:mcmc}
Two shortcomings of \smc are: \emph{(i)} it does not solve the parameter learning problem, and
\emph{(ii)} the quality of the estimates of marginal distributions $p(\Xk) = \int \bar\gamma_K(
\Xk[K]) \rmd \X_{\Ind_K \setminus \Ind_k}$ deteriorates for $k \ll K$ due to the fact that the
particle trajectories degenerate as the particle system evolves (see \eg, \cite{DoucetJ:2011}).
Many methods have been proposed in the literature to address these problems; see \eg
\cite{lindsten2013backward} and the references therein.
Among these, the recently proposed
particle MCMC (PMCMC) framework \cite{andrieu2010particle}, plays a prominent role.
PMCMC algorithms make use of \smc to construct (in general) high-dimensional Markov kernels
that can be used within MCMC. These methods were shown by \cite{andrieu2010particle} to be exact, in the sense that
the apparent particle approximation in the construction of the kernel does
not change its invariant distribution. This property holds for any number of particles
$\Np \geq 2$, \ie, PMCMC does not rely on asymptotics in $\Np$ for correctness.

The fact that the \smc sampler for \pgm{s} presented in Algorithm~\ref{alg:smc} fits under
a general \smc umbrella implies that we can also straightforwardly make use
of this algorithm within PMCMC. This allows us to construct a Markov kernel (indexed by
the number of particles $\Np$) on the space of latent variables of the \pgm, $P_N(\X_{\Ipgas}^\prime, \rmd \X_{\Ipgas})$,
which leaves the full joint distribution $p(\xV)$ invariant.
We do not dwell on the details of the implementation here, but refer instead to \cite{andrieu2010particle}
for the general setup and \cite{LindstenJS:2014} for the specific method that we have used
in the numerical illustration in Section~\ref{sec:experiments}.

PMCMC methods enable blocking of the latent variables of the \pgm in an \mcmc scheme.
Simulating all the latent variables $\X_{\Ipgas}$ jointly is useful since, in general, this will
reduce the autocorrelation when compared to simulating the variables $x_j$ one at a time
\citep{robert1999monte}. However, it is also possible to
employ PMCMC to construct an algorithm in between these two extremes,
a strategy that we believe will be particularly useful in the context of \pgm{s}.
Let $\{ \V^m, \,m \in \set{1}{M} \}$ be a partition of $\V$. 
Ideally, a
Gibbs sampler for the joint distribution $p(\xV)$ could then be constructed by simulating,
using a systematic or a random scan, from the conditional distributions
\begin{align}
  \label{eq:blocking:gibbs}
  p(\X_{\V^m} | \X_{\V\setminus\V^m} )  \text{ for } m = \range{1}{M}.
\end{align}
We refer to this strategy as \emph{partial blocking}, since it amounts to simulating
a subset of the variables, but not necessarily all of them, jointly.
Note that, if we set $M = |\V|$ and $\V^m = \{m\}$ for $m = \range{1}{M}$, this scheme reduces
to a standard 
Gibbs sampler. On the other extreme, with $M = 1$ and $\V^1 = \V$, we get a fully blocked sampler which targets directly
the full joint distribution $p(\xV)$.%

From \eqref{eq:pfactors} it follows that the conditional distributions \eqref{eq:blocking:gibbs} can be expressed as
\begin{equation}
  \label{blocking:conditional}
  p(\X_{\V^m} | \X_{\V\setminus\V^m} ) \propto \prod_{C \in \C^m} \psi_C (\X_C),
\end{equation}
where $\C^m = \{ C \in \C : C \cap \V^m \neq \emptyset \}$.
While it is in general not possible to sample exactly from these conditionals,
we can make use of PMCMC to facilitate a partially blocked Gibbs sampler for a \pgm.
By letting $p( \X_{\V^m} | \X_{\V \setminus \V^m} )$ be the target distribution for the \smc sampler of
Algorithm~\ref{alg:smc}, we can construct a PMCMC kernel $P_N^m$ that leaves the conditional
distribution \eqref{blocking:conditional} invariant. This suggests the following approach:
with $\xV^\prime$ being the current state of the Markov chain, update block $m$ by sampling
\begin{align}
  \X_{\V^m} \sim P_N^m\langle \X_{\V \setminus \V^m}^\prime \rangle(\X_{\V^m}^\prime, \cdot).
\end{align}
Here we have indicated explicitly in the notation that the PMCMC kernel for the conditional
distribution $p( \X_{\V^m} | \X_{\V \setminus \V^m} )$ depends on both
$\X_{\V \setminus \V^m}^\prime$ (which is considered to be fixed throughout the sampling procedure)
and on $\X_{\V^m}^\prime$ (which defines the current state of the PMCMC procedure).

As mentioned above, while being generally applicable, we believe that partial blocking of PMCMC
samplers will be particularly useful for \pgm{s}.
The reason is that we can choose the vertex sets
$\V^m$ for $m = \range{1}{M}$
in order to facilitate simple sequential decompositions of the induced subgraphs.
For instance, it is always possible to choose the partition
in such a way
that all the induced subgraphs are chains.


\section{Experiments}\label{sec:experiments}
In this section we evaluate the proposed \smc sampler on three examples to illustrate the
merits of our approach. Additional details and results are available in the appendix and code to reproduce results can be found in \citep{naessethlsCODE}. We first consider an example from statistical mechanics, the classical XY model, to illustrate the impact of the sequential decomposition. Furthermore, we profile our algorithm with the ``gold standard'' \ais \citep{neal2001annealed} and Annealed Sequential Importance Resampling (\asir\footnote{\asir is a specific instance of the \emph{\smc sampler} by \citep{del2006sequential}, corresponding to \ais with the addition of resampling steps, but to avoid confusion with the proposed method we choose to refer to it as \asir.}) \citep{del2006sequential}.
In the second example we apply the proposed method to the problem of scoring of topic models,
and finally we consider a simple toy model, a Gaussian Markov random field (MRF), which illustrates
that our proposed method has the potential to significantly decrease correlations between samples in
an \mcmc scheme. Furthermore, we provide an \emph{exact} \smc-approximation of the tree-sampler by \citet{hamze2004fields} and thereby extend the scope of this powerful method.%

\subsection{Classical XY model}
\label{sec:xy}
\begin{wrapfigure}{r}{0.5\textwidth}
\begin{center}
\vspace{-30pt}
\includegraphics[width=0.48\textwidth]{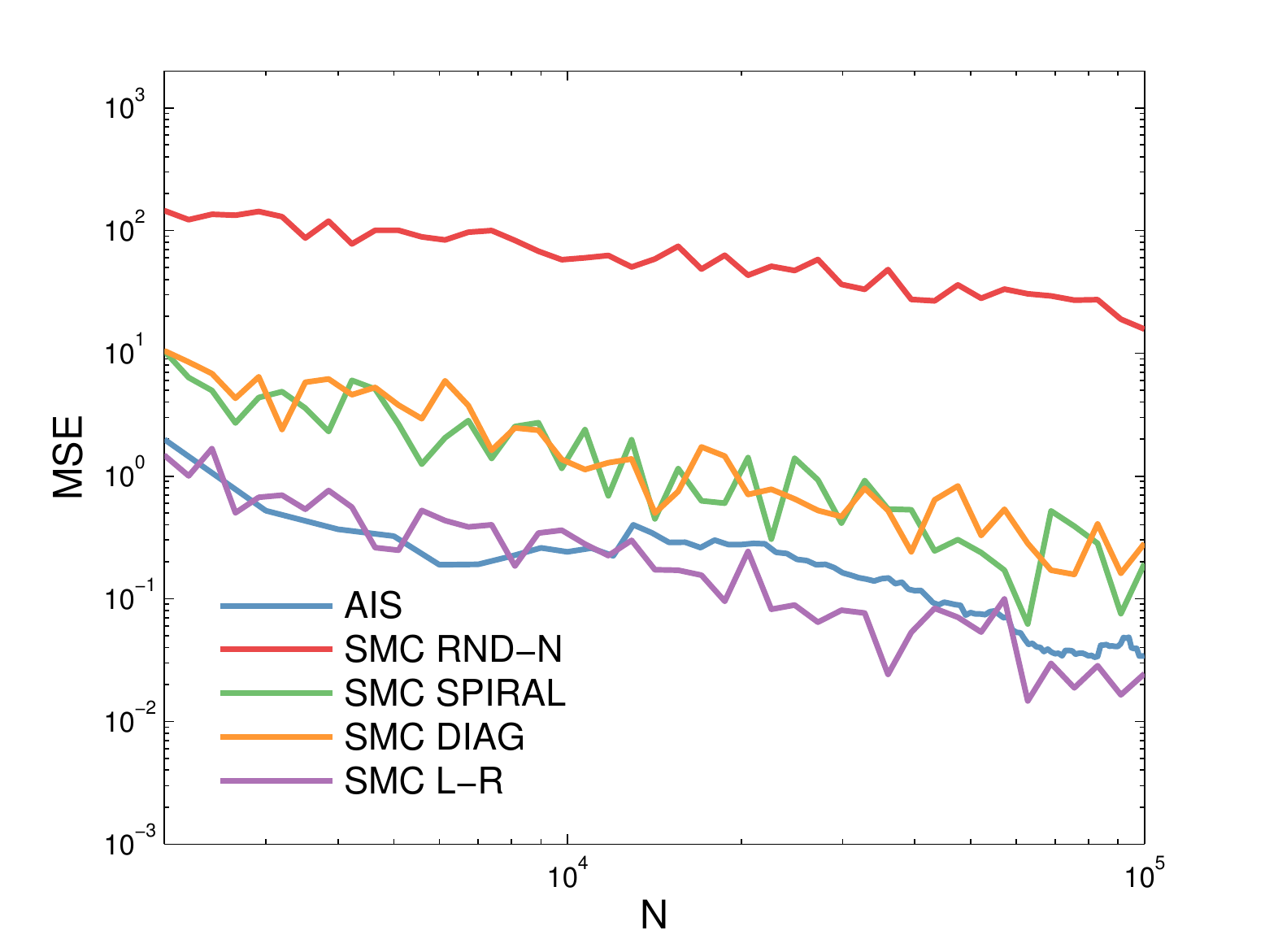}
\vspace{-10pt}
\end{center}
\caption{Mean-squared-errors for sample size $N$ in the estimates of $\log Z$ for \ais and four different orderings in the proposed \smc framework.}
\label{fig:16b11}
\end{wrapfigure}
The classical XY model (see \eg \citep{kosterlitz1973ordering})
is a member in the family of {\em n-vector}
models used in statistical mechanics. It can be seen as a generalization of
the well known Ising model with a two-dimensional electromagnetic
spin. The spin vector
is described by its angle $x \in (-\pi,\pi]$. We will consider square lattices with periodic boundary conditions. The joint \pdf of the classical XY model with equal interaction is given by
\begin{equation}
p(\xV) \propto \mathrm{e}^{\beta \sum_{(i,j) \in \E} \cos (x_i - x_j)},
\end{equation}
where $\beta$ denotes the  inverse temperature.

To evaluate the effect of different sequence orders on the accuracy of the estimates of the
log-normalizing-constant $\log Z$ we ran several experiments on a $16 \times 16$ XY model with
$\beta = 1.1$ (approximately the critical inverse temperature \citep{TomitaO:2002}). For simplicity
we add one node at a time and all factors bridging this node with previously added nodes. Full
adaptation in this case is possible due to the optimal proposal being a von Mises distribution. We
show results for the following cases: \emph{Random neighbour (RND-N)} First node selected randomly
among all nodes, concurrent nodes selected randomly from the set of nodes with a neighbour in
$\Xk[k-1]$. \emph{Diagonal (DIAG)} Nodes added by traversing diagonally ($45^{\circ}$ angle) from
left to right. \emph{Spiral (SPIRAL)} Nodes added spiralling in towards the middle from the
edges. \emph{Left-Right (L-R)} Nodes added by traversing the graph left to right, from top to
bottom.

We also give results of \ais with single-site-Gibbs updates and $1\thinspace000$ annealing distributions linearly spaced from zero to one,
starting from a uniform distribution
(geometric spacing did not yield any improvement over linear spacing for this case). The ``true value'' was estimated using \ais with $10\thinspace000$ intermediate distributions and $5\thinspace000$ importance samples. We can see from the results in Figure~\ref{fig:16b11} that designing a good sequential decomposition for the \smc sampler is important. However, the intuitive and fairly simple choice L-R does give very good results comparable to that of \ais.

Furthermore, we consider a larger size of $64 \times 64$ and evaluate the performance of the L-R ordering compared to \ais and the \asir method. Figure \ref{fig:xy64} displays box-plots of $10$ independent runs. We set $N=10^5$ for
the proposed \smc sampler and then match the computational costs of \ais and \asir with this computational budget. A fair amount of time was spent in tuning the \ais and \asir algorithms; $10\thinspace000$ linear annealing distributions seemed to give best performance in these cases.
\begin{figure}[tb]
  \centering
  \begin{subfigure}{0.31\textwidth}
  \centering
    \includegraphics[width=1.1\textwidth]{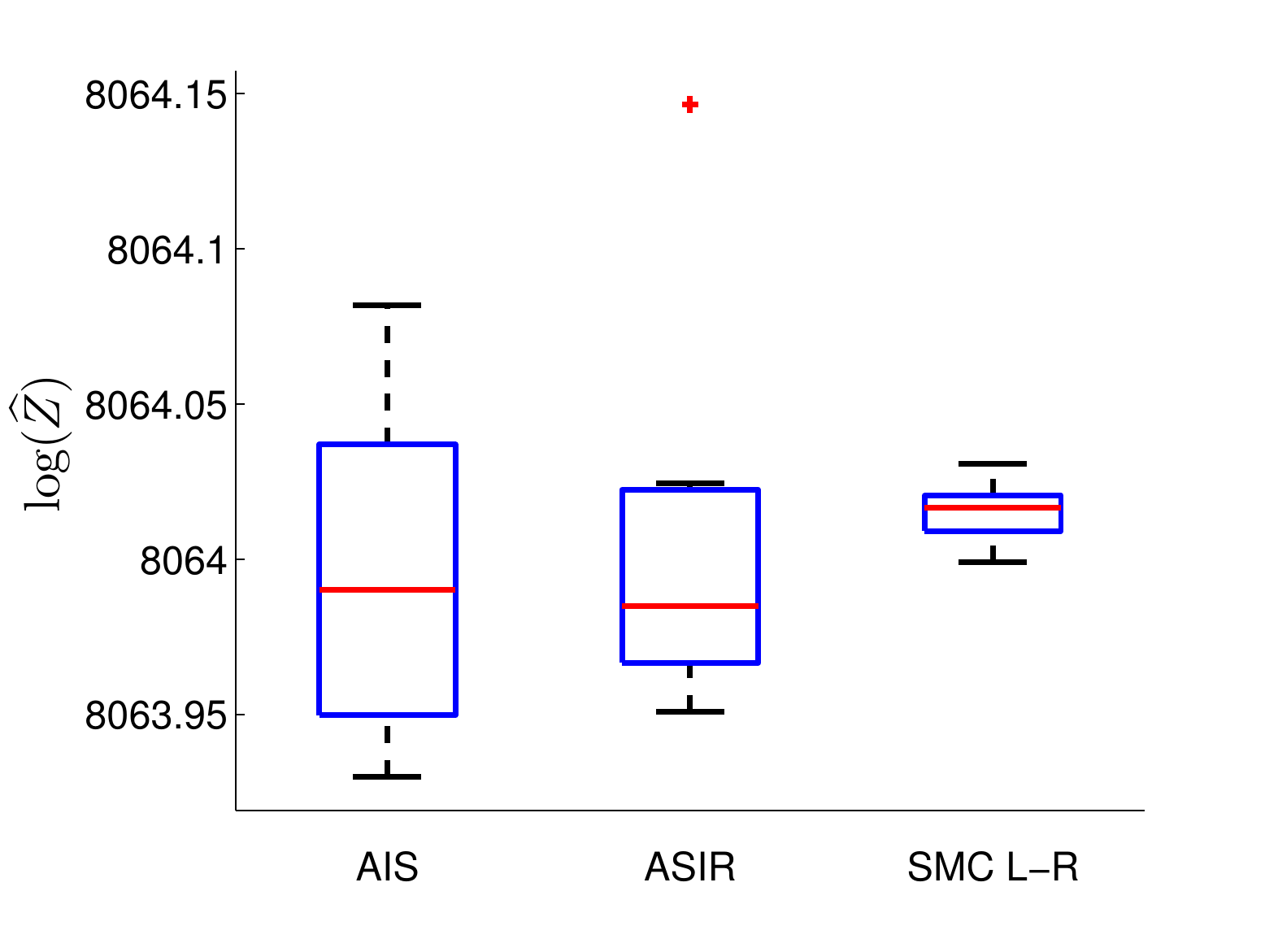}
  \end{subfigure}
  \begin{subfigure}{0.31\textwidth}
  \centering
    \includegraphics[width=1.1\textwidth]{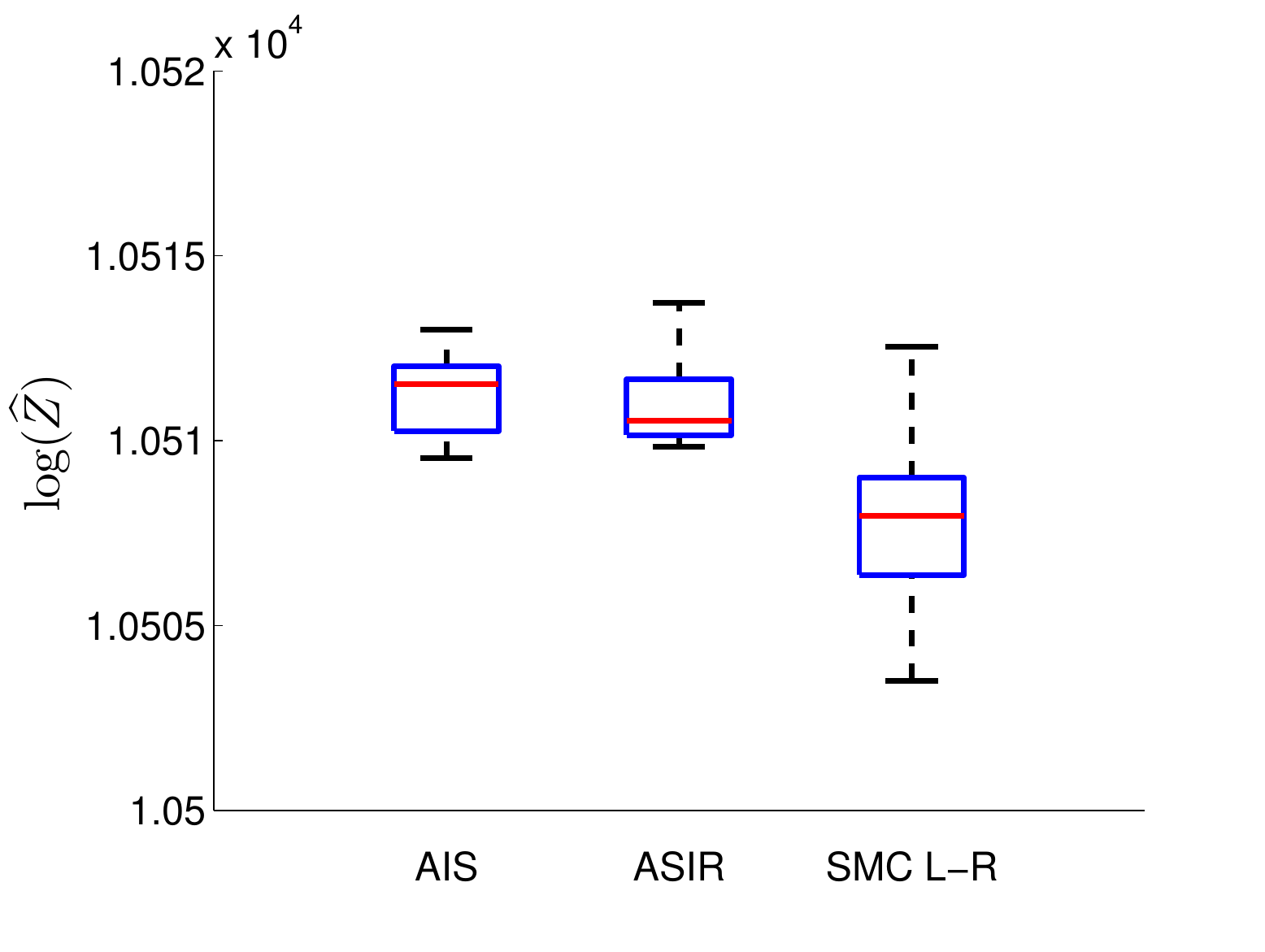}
  \end{subfigure}
  \begin{subfigure}{0.31\textwidth}
  \centering
    \includegraphics[width=1.1\textwidth]{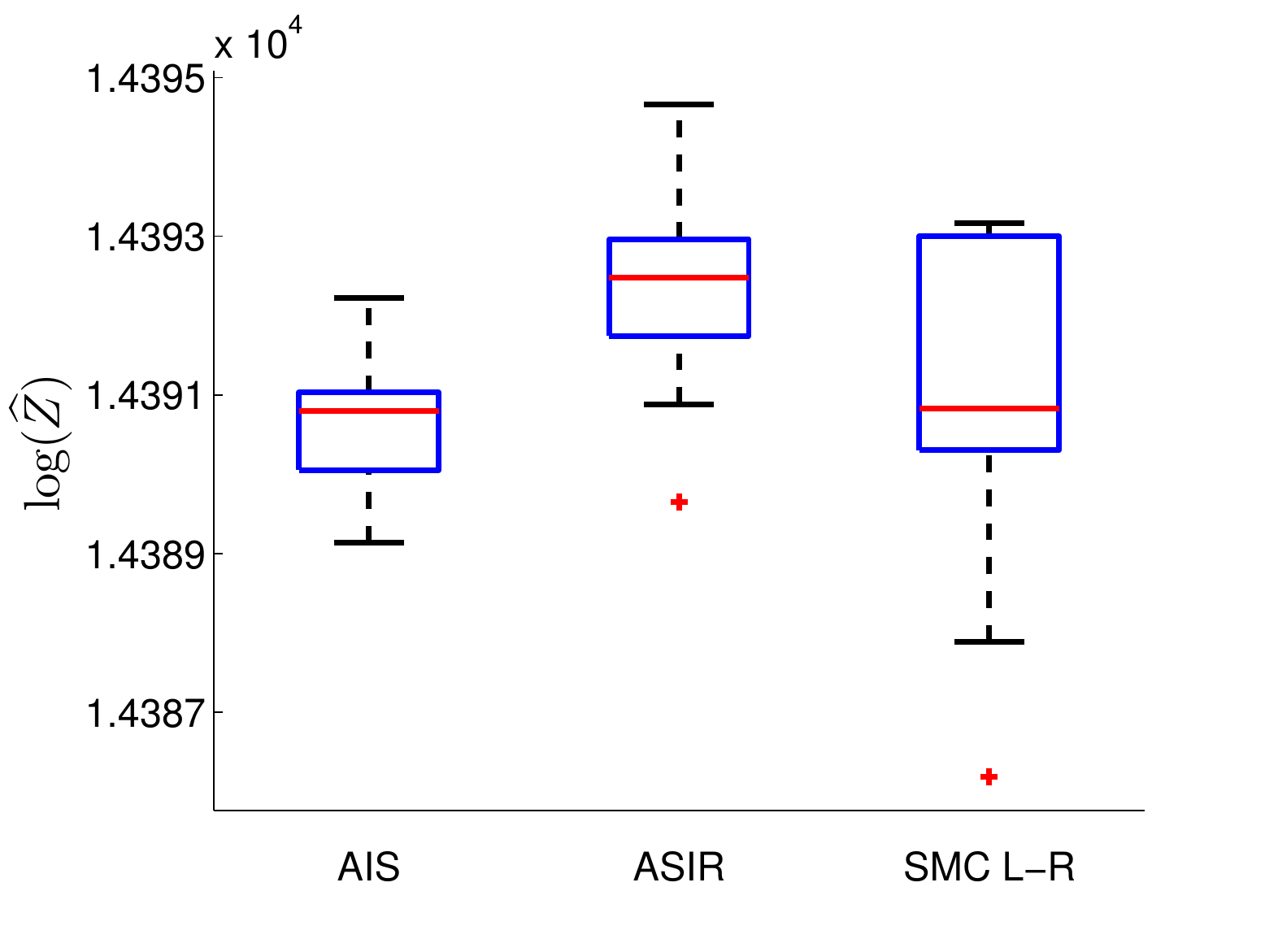}
  \end{subfigure}
  \caption{The logarithm of the estimated partition function for the $64 \times 64$ XY model with inverse temperature $0.5$ (left), $1.1$ (middle) and $1.7$ (right).}
  \label{fig:xy64}
\end{figure}
We can see that the L-R ordering gives results comparable to fairly well-tuned \ais and \asir algorithms;
the ordering of the methods depending on the temperature of the model. One option that does make the \smc algorithm interesting for these types of applications is that it can easily be parallelized over the particles, whereas \ais/\asir has limited possibilities of parallel implementation over the (crucial) annealing steps.

\subsection{Likelihood estimation in topic models}
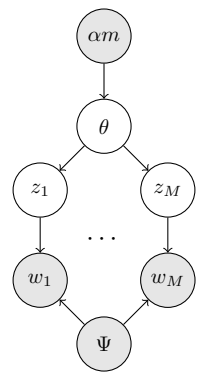
\begin{wrapfigure}{r}{0.25\textwidth}
\centering
\vspace{-10pt}
\begin{tikzpicture}[->,node distance=1.5cm,main node/.style={circle,draw,minimum width=0.9cm,scale=0.8},param node/.style={circle,draw,minimum width=0.9cm,scale=0.8,fill=gray!20}]
  \node[param node] (1) {$\alpha m$};
  \node[main node] (2) [below of=1] {$\theta$};
  \node[main node] (3) [below left of=2] {$z_1$};
  \node[param node] (4) [below of=3] {$w_1$};
  \node (5) [below of=2] {$\cdots$};
  \node[main node] (6) [below right of=2] {$z_M$};
  \node[param node] (7) [below of=6] {$w_M$};
  \node[param node] (8) [below left of=7] {$\Psi$};

  \path (1) edge (2);
  \path (2) edge (3);
  \path (3) edge (4);
  \path (2) edge (6);
  \path (6) edge (7);
  \path (8) edge (4);
  \path (8) edge (7);
\end{tikzpicture}
\caption{\lda as graphical model.}
\vspace{-10pt}
\label{fig:ldaPGM}
\end{wrapfigure}
Topic models such as Latent Dirichlet Allocation (\lda) \citep{blei2003latent}
are popular models for reasoning about large text corpora.
Model evaluation is often conducted by computing the likelihood of held-out documents
w.r.t.\ a learnt model. However, this is a challenging problem on its own---which has received much recent interest \citep{wallach2009evaluation, buntine2009estimating, scott2013arecursive}---since it essentially corresponds to computing the partition function of a graphical model; see Figure~\ref{fig:ldaPGM}.
The \smc procedure of Algorithm~\ref{alg:smc} can used to solve this problem by defining a sequential decomposition of the graphical model.
In particular, we consider the decomposition corresponding to first including the node $\theta$ and then,
 subsequently, introducing
the nodes $z_{1}$ to $z_M$ in any 
order. Interestingly, if we then make use of a Rao-Blackwellization
over the variable $\theta$, the \smc sampler of Algorithm~\ref{alg:smc} reduces exactly to
a method that has previously been proposed for this specific problem \citep{scott2013arecursive}.
In \citep{scott2013arecursive}, the method is derived by reformulating the model in terms of its sufficient
statistics and phrasing this as a particle learning problem; here we obtain the same
procedure as a special case of the general \smc algorithm operating on the original model.

We use the same data and learnt models as \citet{wallach2009evaluation}, \ie 20 newsgroups, and PubMed Central abstracts (PMC). We compare with the Left-Right-Sequential (\lrs) sampler \citep{buntine2009estimating},
which is an improvement over the method proposed by \citet{wallach2009evaluation}.
Results on simulated and real data experiments are provided in Figure~\ref{fig:lda}. For the simulated example (Figure~\ref{fig:lda:sim}), we use a small model with 10 words and 4 topics to be able to compute the exact log-likelihood.
We keep the number of particles in the \smc algorithm equal to the number of Gibbs steps in \lrs; this means \lrs is about an order-of-magnitude more computationally demanding than the \smc method. 
Despite the fact that the \smc sampler uses only about a tenth of the computational time of the \lrs
sampler, it performs significantly better in terms of estimator variance.
\begin{figure}[tb]
  \centering
  \begin{subfigure}{0.32\textwidth}
  \centering
    \includegraphics[width=1.0\textwidth]{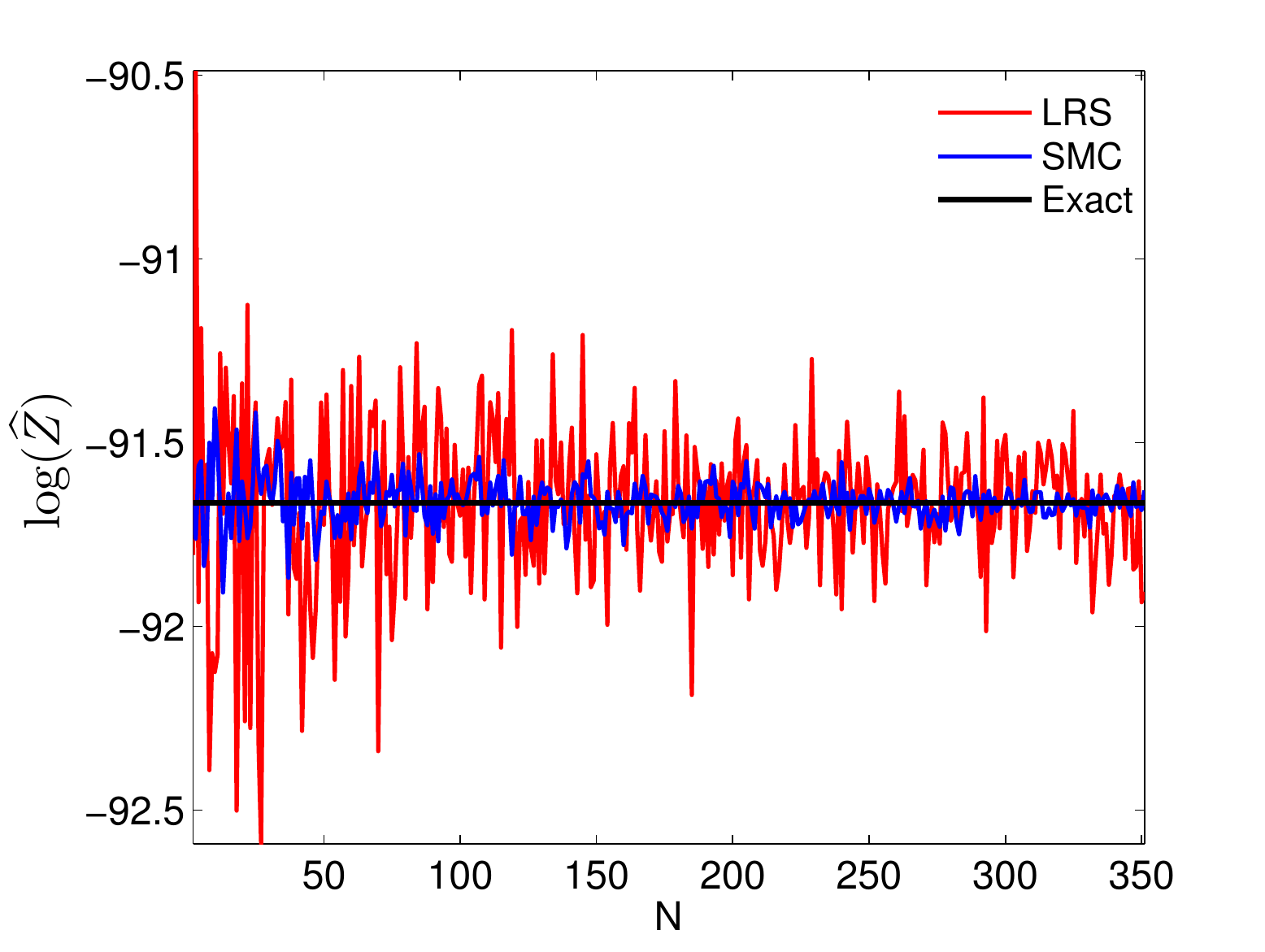}
    \caption{Small simulated example.}\label{fig:lda:sim}
  \end{subfigure}
  \begin{subfigure}{0.32\textwidth}
  \centering
    \includegraphics[width=1.0\textwidth]{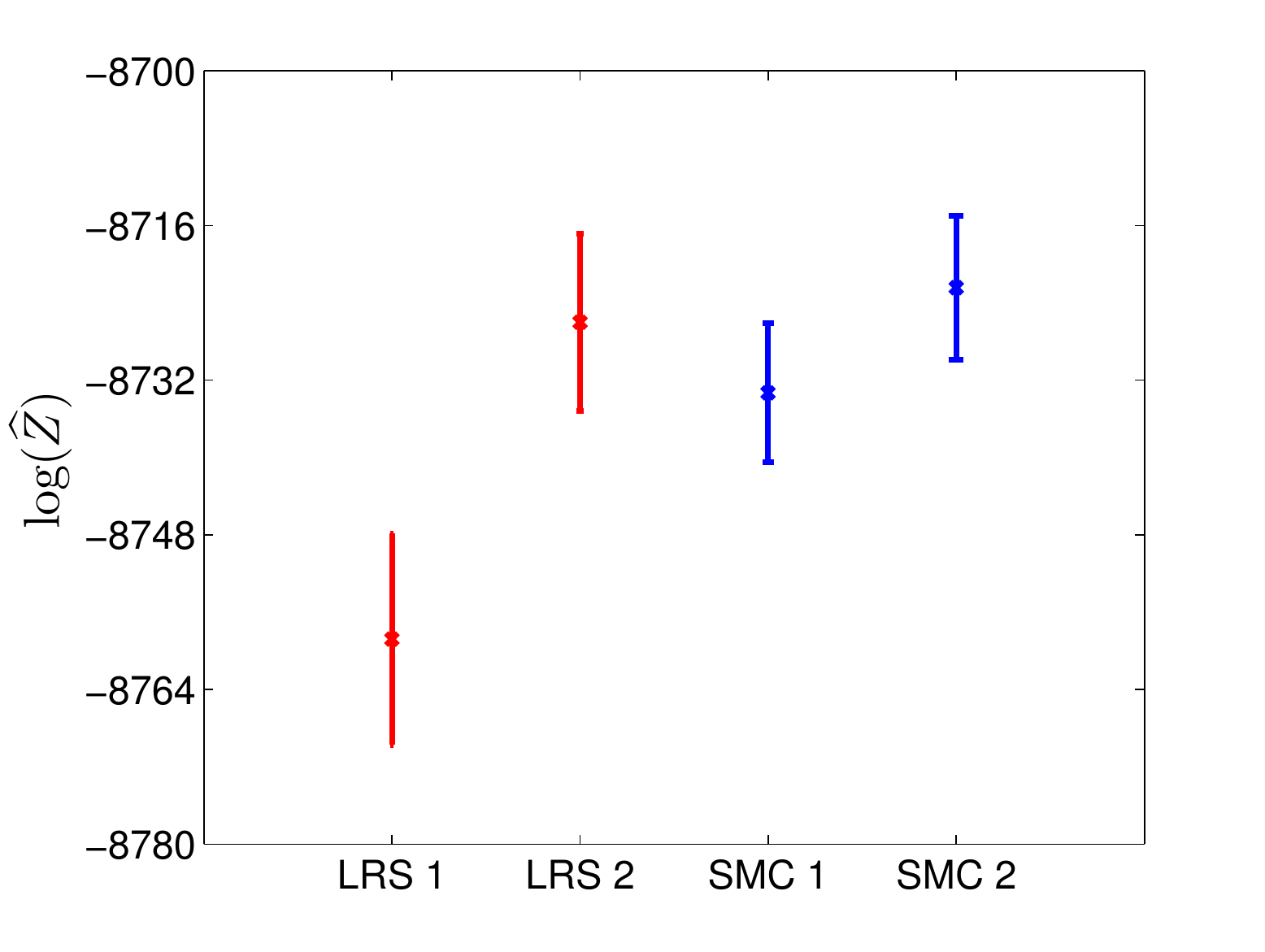}
    \caption{PMC.}
  \end{subfigure}
  \begin{subfigure}{0.32\textwidth}
  \centering
    \includegraphics[width=1.0\textwidth]{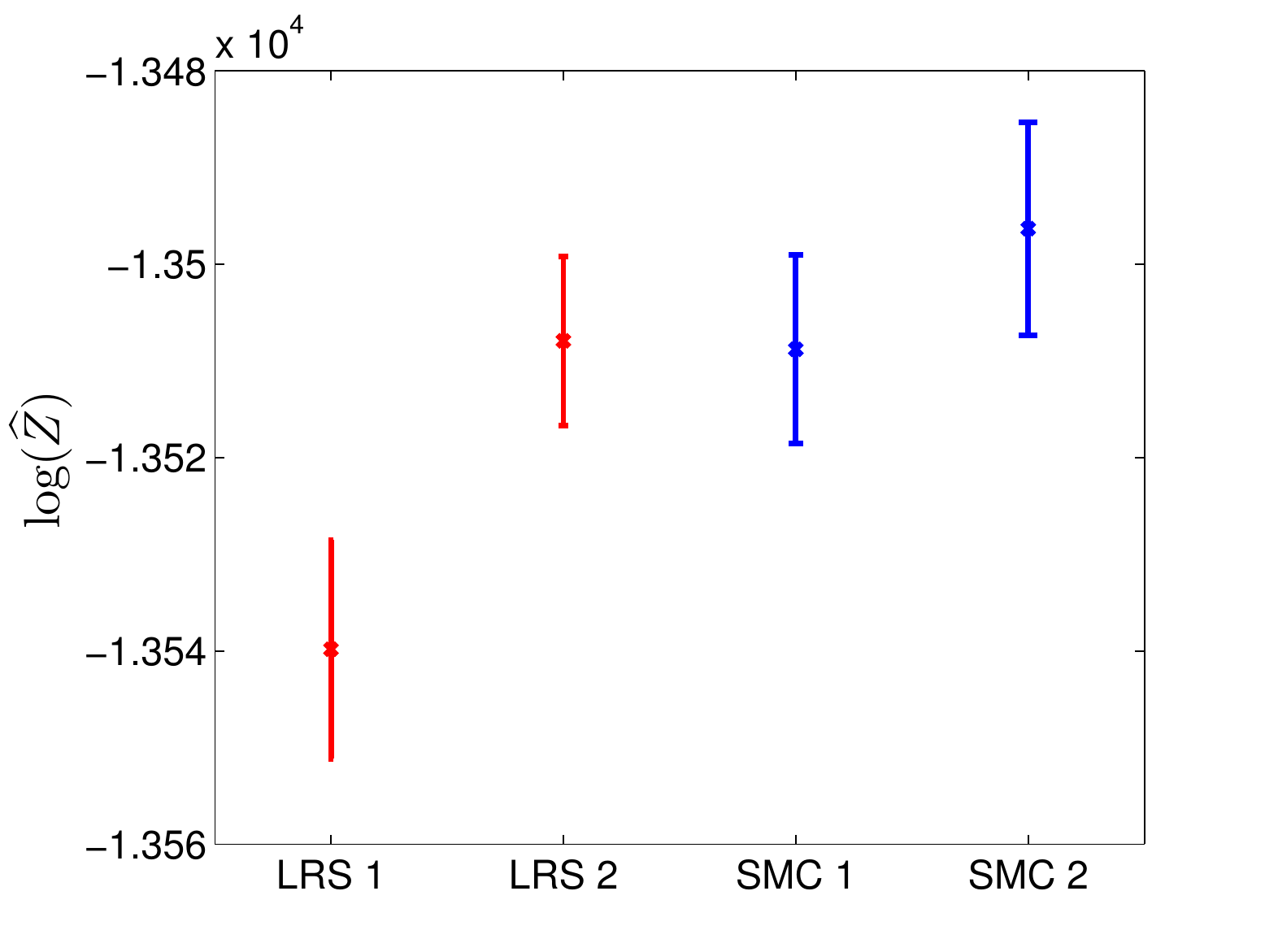}
    \caption{20 newsgroups.}
  \end{subfigure}
  \caption{Estimates of the log-likelihood of heldout documents for various datasets.}
  \label{fig:lda}
\end{figure}
The other two plots show results on real data with $10$ held-out documents for each dataset. For a fixed number of Gibbs steps
we choose the number of particles for each document to make the computational cost approximately equal. 
Run \#2 has twice the number of particles/samples as in run \#1. We show the mean of $10$ runs and
error-bars estimated using bootstrapping with $10\thinspace000$ samples. Computing the logarithm of $\hat Z$ introduces a negative bias, which means larger values of $\log \hat Z$ typically implies more accurate results. The results on real data do
not show the drastic improvement we see in the simulated example, which could be due to degeneracy problems for long documents. An interesting approach that could improve results would be to use an \smc algorithm tailored to discrete distributions, \eg \citet{fearnhead2003online}.

\subsection{Gaussian MRF}
Finally, we consider a simple toy model to illustrate how the \smc sampler of Algorithm~\ref{alg:smc}
can be incorporated in PMCMC sampling. 
We simulate data from a zero mean Gaussian  $10 \times 10$ lattice \mrf with observation and interaction standard deviations of $\sigma_i=1$ and $\sigma_{ij} = 0.1$ respectively. We use the proposed \smc algorithm together with the \pmcmc method by \citet{LindstenJS:2014}.
We compare this with standard Gibbs sampling and the tree sampler by \citet{hamze2004fields}.

\begin{wrapfigure}{r}{0.4\textwidth}
  \vspace{-12pt}
     \includegraphics[width=0.4\textwidth]{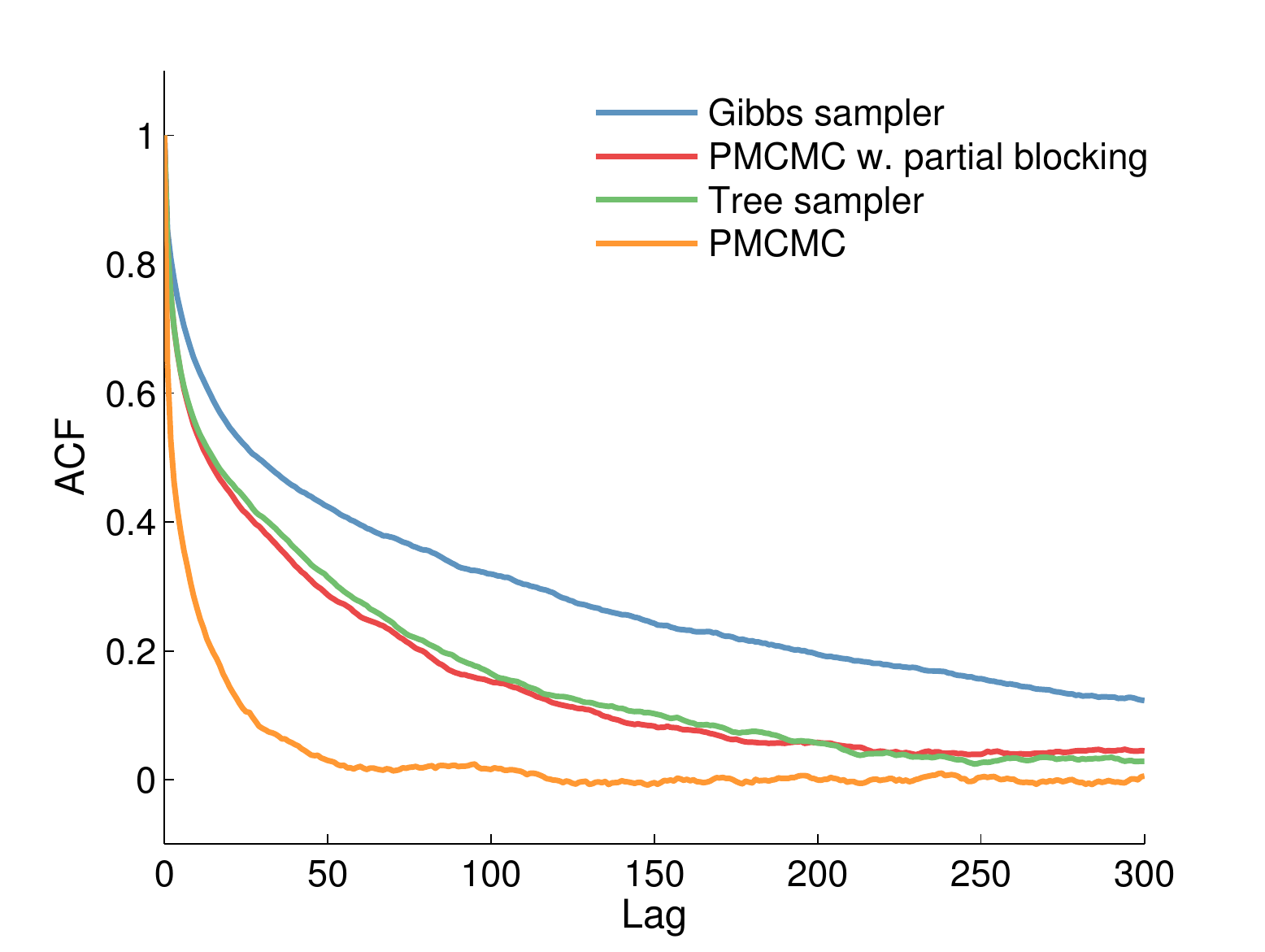}
     \caption{The empirical ACF for Gibbs sampling, PMCMC, PMCMC with partial blocking, and tree sampling.}\label{fig:gmrfacf}
\end{wrapfigure}
We use a moderate number of $\Np = 50$ particles in the PMCMC sampler (recall that it admits the correct invariant
distribution for any $\Np \geq 2$).
In Figure~\ref{fig:gmrfacf} we can see the empirical autocorrelation funtions (ACF) centered around the true posterior mean for variable $x_{82}$ (selected randomly from among $\X_\V$; similar results hold for all the variables of the model).  Due to the strong interaction between the
latent variables, the samples generated by the standard Gibbs sampler are strongly correlated.
Tree-sampling and PMCMC with partial blocking show nearly identical gains compared to Gibbs. This is interesting, since it suggest that simulating from the \smc-based PMCMC kernel can be almost as efficient as exact simulation, even using a moderate number of particles. Indeed, PMCMC with partial blocking can be viewed as an \emph{exact} \smc-approximation
of the tree sampler, extending the scope of tree-sampling beyond discrete and Gaussian models.
The fully blocked PMCMC algorithm achieves the best ACF, dropping off to zero considerably faster than for the other methods. This is not surprising since this sampler simulates all the latent variables jointly which reduces the autocorrelation, in particular when the latent variables are strongly dependent. However, it should be noted
that this method also has the highest computational cost per iteration.


\section{Conclusion}\label{sec:discussion}
We have proposed a new framework for inference in \pgm{s} using \smc and illustrated it on three examples. These examples show that it can be a viable alternative to standard methods used for inference and partition function estimation problems. An interesting avenue for future work is combining our proposed methods with \ais, to see if we can improve on both.

\subsubsection*{Acknowledgments}
We would like to thank Iain Murray for his kind and very prompt help in providing the data for the
LDA example. This work was supported by the projects: \emph{Learning of complex dynamical systems}
(Contract number: 637-2014-466) and \emph{Probabilistic modeling of dynamical systems} (Contract
number: 621-2013-5524), both funded by the Swedish Research Council.

\appendix
\section{Supplementary material}
  This appendix contains additional information on the experiments in the main
  paper \citep{NaessethLS:2014} as well as a simple and direct proof of the unbiasedness of the
  partition function estimator $\widehat Z^N_k$, stated in the main manuscript.
  It should be noted, however, that this result is not new. It has previously been established in a general setting
  by \citet[Proposition~7.4.1]{DelMoral:2004} and, additionally, by \citet{PittSGK:2012} who provide a more accessible proof
  for the special case of state-space models.
  Our proof is similar to that of \citet{PittSGK:2012}, but generalized to the \pgm setting that we consider.

\subsection{Experiments}
\subsubsection{Classical XY model}
The classical XY model, see \eg \citep{kosterlitz1973ordering} and references therein, is a member in the family of {\em n-vector}
models used in statistical mechanics. It can be seen as a generalization of
the well known Ising model with a two-dimensional electromagnetic
spin. The spin vector
is described by its angle $x \in (-\pi,\pi]$. We will consider a square lattice with periodic
boundary conditions, \ie the first and last row/columns are connected. The individual sites are described by their spin angle.

The full joint \pdf of the classical XY model is given by
\begin{equation}
p(\xV) \propto e^{-\beta H(\xV)},
\end{equation}
where $\beta$ is the inverse temperature and $H(\xV)$---the
Hamiltonian---is a sum of pair-wise interaction described by
\begin{equation}
  H(\xV) = -\sum_{(i,j)\in \mathcal{E}} J_{ij} \cos{(x_i - x_j)},
\end{equation}
where the $J_{ij}$'s are parameters describing interactions between
the different sites. For simplicity we set $J_{ij} = J = 1$ and estimate the partition function for several sizes and $\beta$. 

In our sequence of target distributions we add one variable at a time and all associated factors. A simple example where we alternate left-right, right-left, can be seen in Figure~\ref{fig:blockorder}. To be specific we choose our sequence of intermediate target distributions as
\begin{align}
\gamma_k (\Xk) &\propto \gamma_{k-1}(\Xk[k-1]) \prod_{i \in \mathcal{N}_k} e^{\beta J_{ki} \cos{(x_k -x_i)}} \nonumber \\
&\propto \gamma_{k-1}(\Xk[k-1]) e^{\kappa(\Xk[k-1]) \cos{(x_k -\mu(\Xk[k-1]))}},
\end{align}
where $\mathcal{N}_k = \{i:(k,i) \in \E\}  \cap \Ind_{k-1}$ denotes the set of neighbours to variable $k$ in $\Ind_{k-1}$. The quantities $\mu(\Xk[k-1])$ and $\kappa(\Xk[k-1])$ follow from elementary trigonometric operations (sum of cosines). From the above expression we note that, conditionally on $\Xk[k-1]$, the variable $x_k$ is von Mises distributed under $\bar\gamma_k$,
with $\Xk[k-1]$-dependent mean $\mu$ and dispersion $\kappa$. This implies that we can employ full adaption of the proposed \smc sampler.
This is accomplished by choosing the aforementioned von Mises distribution as proposal distribution
$r_k (x_k | \Xk[k-1])$ and by choosing the corresponding normalizing constants $\nu(\Xk[k-1]) = 2
\pi I_0 (\kappa(\Xk[k-1]))$ (where $I_0$ is the modified Bessel function of order $0$) as adjustment
weights. 
We use the fully adapted \smc sampler to estimate the partition function of the classical XY model.

We consider four different orderings of the nodes:
\begin{description}
\item[Random neighbour] The first node is selected randomly among all nodes, concurrent nodes are
  then selected randomly from the set of nodes with a neighbour in $\Xk[k-1]$.

\item[Diagonal] The nodes are added by traversing from left to right with $45^{\circ}$, see Figure~\ref{fig:diag}. 

\item[Spiral] The nodes are added spiralling in towards the middle from the edges, see Figure~\ref{fig:spiral}.

\item[Left-Right] The nodes are added by traversing the graph from left to right, from top to bottom, see Figure~\ref{fig:lr}.
\end{description}
See illustrations of the node orderings displayed in Figure~\ref{fig:xy:order} for a $3 \times 3$ example, numbers display at what iteration the node is added.
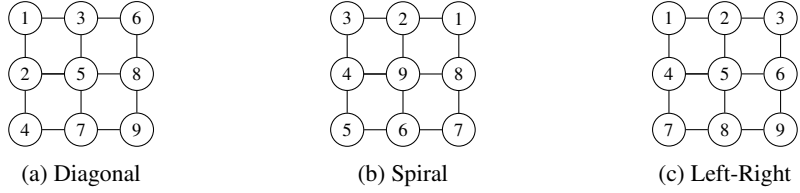
\begin{figure}
\centering
\begin{subfigure}[b]{0.3\textwidth}
\centering
\begin{tikzpicture}[node distance=1.06cm,main node/.style={circle,draw, scale=0.7}]
  \node[main node] (1) {1};
  \node[main node] (2) [right of=1] {3};
  \node[main node] (3) [right of=2] {6};
  \node[main node] (4) [below of=1] {2};
  \node[main node] (5) [right of=4] {5};
  \node[main node] (6) [right of=5] {8};
  \node[main node] (7) [below of=4] {4};
  \node[main node] (8) [right of=7] {7};
  \node[main node] (9) [right of=8] {9};

  \path (1) edge (2);
  \path (2) edge (3);
  \path (1) edge (4);
  \path (2) edge (5);
  \path (3) edge (6);
  \path (4) edge (5);
  \path (5) edge (6);
  \path (4) edge (5);
  \path (4) edge (7);
  \path (5) edge (8);
  \path (6) edge (9);
  \path (7) edge (8);
  \path (8) edge (9);
\end{tikzpicture}
\caption{Diagonal}
\label{fig:diag}
\end{subfigure}
\begin{subfigure}[b]{0.3\textwidth}
\centering
\begin{tikzpicture}[node distance=1.06cm,main node/.style={circle,draw, scale=0.7}]
  \node[main node] (1) {3};
  \node[main node] (2) [right of=1] {2};
  \node[main node] (3) [right of=2] {1};
  \node[main node] (4) [below of=1] {4};
  \node[main node] (5) [right of=4] {9};
  \node[main node] (6) [right of=5] {8};
  \node[main node] (7) [below of=4] {5};
  \node[main node] (8) [right of=7] {6};
  \node[main node] (9) [right of=8] {7};

  \path (1) edge (2);
  \path (2) edge (3);
  \path (1) edge (4);
  \path (2) edge (5);
  \path (3) edge (6);
  \path (4) edge (5);
  \path (5) edge (6);
  \path (4) edge (5);
  \path (4) edge (7);
  \path (5) edge (8);
  \path (6) edge (9);
  \path (7) edge (8);
  \path (8) edge (9);
\end{tikzpicture}
\caption{Spiral}
\label{fig:spiral}
\end{subfigure}
\begin{subfigure}[b]{0.3\textwidth}
\centering
\begin{tikzpicture}[node distance=1.06cm,main node/.style={circle,draw, scale=0.7}]
  \node[main node] (1) {1};
  \node[main node] (2) [right of=1] {2};
  \node[main node] (3) [right of=2] {3};
  \node[main node] (4) [below of=1] {4};
  \node[main node] (5) [right of=4] {5};
  \node[main node] (6) [right of=5] {6};
  \node[main node] (7) [below of=4] {7};
  \node[main node] (8) [right of=7] {8};
  \node[main node] (9) [right of=8] {9};

  \path (1) edge (2);
  \path (2) edge (3);
  \path (1) edge (4);
  \path (2) edge (5);
  \path (3) edge (6);
  \path (4) edge (5);
  \path (5) edge (6);
  \path (4) edge (5);
  \path (4) edge (7);
  \path (5) edge (8);
  \path (6) edge (9);
  \path (7) edge (8);
  \path (8) edge (9);
\end{tikzpicture}
\caption{Left-Right}
\label{fig:lr}
\end{subfigure}
\caption{Illustration of some of the different orderings considered in the XY model.}\label{fig:xy:order}
\end{figure}

\subsubsection{Evaluation of topic models}
Here we present some additional results (Figure~\ref{fig:ldaextra}) on the synthetic example for
various settings of the number of topics ($T$) and words ($W$). LRS $1$ and LRS $2$ has $10$ and
$20$ samples, respectively. The number of particles where set to give comparable computational
complexity.
\begin{figure}[h!]
  \centering
  \begin{subfigure}{0.4\textwidth}
    \includegraphics[width=\textwidth]{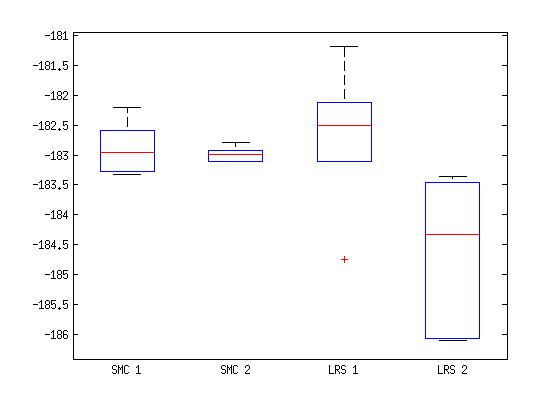}
    \caption{$T=100,W=20$.}\label{fig:lda1}
  \end{subfigure}%
  ~
  \begin{subfigure}{0.4\textwidth}
    \includegraphics[width=\textwidth]{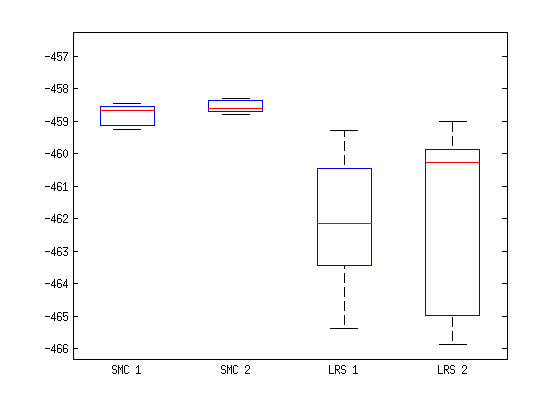}
    \caption{$T=100,W=50$.}\label{fig:lda2}
  \end{subfigure}
  
  \begin{subfigure}{0.4\textwidth}
    \includegraphics[width=\textwidth]{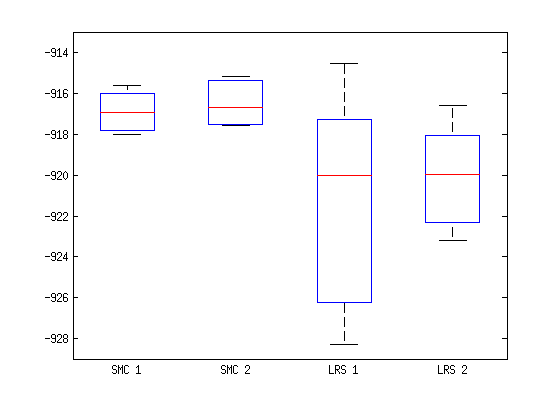}
    \caption{$T=100,W=100$.}\label{fig:lda3}
  \end{subfigure}%
  ~
  \begin{subfigure}{0.4\textwidth}
    \includegraphics[width=\textwidth]{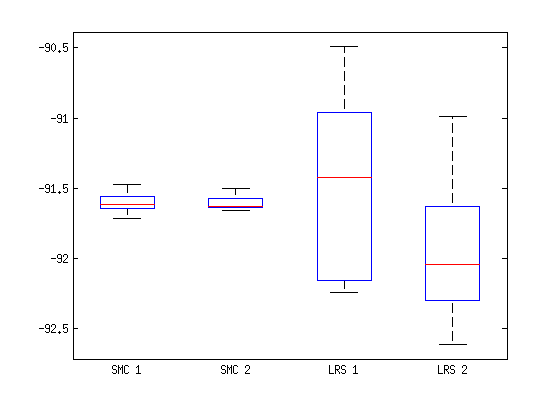}
    \caption{$T=200,W=10$.}\label{fig:lda4}
  \end{subfigure}
  \caption{Estimates of the log-likelihood of a synthetic LDA model.}\label{fig:ldaextra}
\end{figure}

\subsubsection{Gaussian Markov random field}
Consider a square lattice Gaussian Markov random field (MRF) of size $10 \times 10$, given by the relation
\begin{equation}
p(\X_\V, \Y_\V) \propto \prod_{i \in \V} e^{\frac{1}{2\sigma_i^2} (x_i
  - y_i)^2} \prod_{(i,j)
  \in \E} e^{\frac{1}{2\sigma_{ij}^2} (x_i - x_j)^2},
\end{equation}
with latent variables $\X_\V = \set{x_1}{x_{100}}$ and measurements $\Y_\V = \set{y_1}{y_{100}}$.
The graphical representation of the latent variables in this model is shown in
Figure~\ref{fig:blockorder}.

The measurements $\Y_\V$ where simulated from the model with $\sigma_i = 1$ and $\sigma_{ij} = 0.1$.
Given these measurement, we seek the posterior distribution $p(\X_\V \mid \Y_\V)$.  We run four
different \mcmc samplers to simulate from this distribution; the proposed (fully blocked) \pgas, the
proposed \pgas with partial blocking, a standard one-at-a-time Gibbs sampler, and the tree-sampler
proposed by \citet{hamze2004fields}. For the \pgas algorithms we use $N=50$ particles.  The
tree-sampler exploits the fact that the model is Gaussian and it can thus not be used for arbitrary
(non-Gaussian or non-discrete) graphs.  By partitioning the graph into disjoint trees (in our case,
chains) for which exact inference is possible, the tree-sampler implements an ``ideal'' partially
blocked Gibbs sampler.  \pgas with partial blocking can thus be seen as an \smc-based version of the
tree-sampler.  See Figure~\ref{fig:blockorder} for the ordering in the \pgas algorithm and the
blocking used for tree-sampling and \pgas with partial blocking, corresponding to a partition of the
graph into two chains. The variables are numbered $1,\ldots,100$ from top to bottom, left to right and $\Ind_k$ is taken as the $k$ first indices of $\Ind_K = \{1,\ldots,10, 20, 19, \ldots, 11, 21, 22, \ldots, 100, 99, \ldots, 91\}$. This ordering gives results very similar to that of the Left-Right ordering explained above.

\begin{figure}[h!]
  \centering
  \begin{subfigure}{0.4\textwidth}
    \resizebox{0.9\textwidth}{!}{%
    \tikzstyle{edge} = [-,thick]
\tikzstyle{arrw} = [very thick,shorten <=2pt,shorten >=2pt]
\tikzstyle{var} = [draw,circle,inner sep=0,minimum width=0.5cm]
\tikzstyle{obs} = [draw,circle,inner sep=0,minimum width=0.5cm, fill=black!20]
  \begin{tikzpicture}[>=stealth,node distance=0.6cm]
    \begin{scope}
      \foreach \x in {0,1,2,3,4,5,6,7,8} {
        \foreach \y in {9} {
          \node at (\x,\y) (x\x\y) [var] {$\Rightarrow$};
        }
      }
      \foreach \x in {9} {
        \foreach \y in {9} {
          \node at (\x,\y) (x\x\y) [var] {$\Downarrow$};
        }
      }
      \foreach \x in {1,2,3,4,5,6,7,8,9} {
        \foreach \y in {8} {
          \node at (\x,\y) (x\x\y) [var] {$\Leftarrow$};
        }
      }
      \foreach \x in {0} {
        \foreach \y in {8} {
          \node at (\x,\y) (x\x\y) [var] {$\Downarrow$};
        }
      }
      \foreach \x in {0,1,2,3,4,5,6,7,8} {
        \foreach \y in {7} {
          \node at (\x,\y) (x\x\y) [var] {$\Rightarrow$};
        }
      }
      \foreach \x in {9} {
        \foreach \y in {7} {
          \node at (\x,\y) (x\x\y) [var] {$\Downarrow$};
        }
      }
      \foreach \x in {1,2,3,4,5,6,7,8,9} {
        \foreach \y in {6} {
          \node at (\x,\y) (x\x\y) [var] {$\Leftarrow$};
        }
      }
      \foreach \x in {0} {
        \foreach \y in {6} {
          \node at (\x,\y) (x\x\y) [var] {$\Downarrow$};
        }
      }
      \foreach \x in {0,1,2,3,4,5,6,7,8} {
        \foreach \y in {5} {
          \node at (\x,\y) (x\x\y) [var] {$\Rightarrow$};
        }
      }
      \foreach \x in {9} {
        \foreach \y in {5} {
          \node at (\x,\y) (x\x\y) [var] {$\Downarrow$};
        }
      }
      \foreach \x in {1,2,3,4,5,6,7,8,9} {
        \foreach \y in {4} {
          \node at (\x,\y) (x\x\y) [var] {$\Leftarrow$};
        }
      }
      \foreach \x in {0} {
        \foreach \y in {4} {
          \node at (\x,\y) (x\x\y) [var] {$\Downarrow$};
        }
      }
      \foreach \x in {0,1,2,3,4,5,6,7,8} {
        \foreach \y in {3} {
          \node at (\x,\y) (x\x\y) [var] {$\Rightarrow$};
        }
      }
      \foreach \x in {9} {
        \foreach \y in {3} {
          \node at (\x,\y) (x\x\y) [var] {$\Downarrow$};
        }
      }
      \foreach \x in {1,2,3,4,5,6,7,8,9} {
        \foreach \y in {2} {
          \node at (\x,\y) (x\x\y) [var] {$\Leftarrow$};
        }
      }
      \foreach \x in {0} {
        \foreach \y in {2} {
          \node at (\x,\y) (x\x\y) [var] {$\Downarrow$};
        }
      }
      \foreach \x in {0,1,2,3,4,5,6,7,8} {
        \foreach \y in {1} {
          \node at (\x,\y) (x\x\y) [var] {$\Rightarrow$};
        }
      }
      \foreach \x in {9} {
        \foreach \y in {1} {
          \node at (\x,\y) (x\x\y) [var] {$\Downarrow$};
        }
      }
      \foreach \x in {0,1,2,3,4,5,6,7,8,9} {
        \foreach \y in {0} {
          \node at (\x,\y) (x\x\y) [var] {$\Leftarrow$};
        }
      }
      \foreach \x in {0,1,2,3,4,5,6,7,8} {
        \pgfmathtruncatemacro\xend{\x+1}
        \foreach \y in {0,1,2,3,4,5,6,7,8,9} {
          \draw[edge] (x\x\y) -- (x\xend\y);
        }
      }
      \foreach \x in {0,1,2,3,4,5,6,7,8,9} {
        \foreach \y in {0,1,2,3,4,5,6,7,8} {
          \pgfmathtruncatemacro\yend{\y+1}
          \draw[edge] (x\x\y) -- (x\x\yend) {};
        }
      }
    \end{scope}
  \end{tikzpicture}
    }
    \label{fig:order}
  \end{subfigure}%
  ~
  \begin{subfigure}{0.4\textwidth}
    \resizebox{0.9\textwidth}{!}{%
\tikzstyle{edge} = [-,thick]
\tikzstyle{arrw} = [very thick,shorten <=2pt,shorten >=2pt]
\tikzstyle{var} = [draw,circle,inner sep=0,minimum width=0.5cm]
\tikzstyle{obs} = [draw,circle,inner sep=0,minimum width=0.5cm, fill=black!20]
 \begin{tikzpicture}[>=stealth,node distance=0.6cm]
    \begin{scope}
      \foreach \x in {0,1,2,3,4,5,6,7,8,9} {
        \foreach \y in {9} {
          \node at (\x,\y) (x\x\y) [var] {};
        }
      }
      \foreach \x in {0} {
        \foreach \y in {0,1,2,3,4,5,6,7,8} {
          \node at (\x,\y) (x\x\y) [var] {};
        }
      }
      \foreach \x in {0,1,2,3,4,5,6,7,8} {
        \foreach \y in {0} {
          \node at (\x,\y) (x\x\y) [var] {};
        }
      }
     \foreach \x in {8} {
        \foreach \y in {1,2,3,4,5,6,7} {
          \node at (\x,\y) (x\x\y) [var] {};
        }
      }
      \foreach \x in {2,3,4,5,6,7,8} {
        \foreach \y in {7} {
          \node at (\x,\y) (x\x\y) [var] {};
        }
      }
      \foreach \x in {2} {
        \foreach \y in {2,3,4,5,6} {
          \node at (\x,\y) (x\x\y) [var] {};
        }
      }     
      \foreach \x in {3,4,5,6} {
        \foreach \y in {2} {
          \node at (\x,\y) (x\x\y) [var] {};
        }
      }  
      \foreach \x in {6} {
        \foreach \y in {3,4,5} {
          \node at (\x,\y) (x\x\y) [var] {};
        }
      }   
      \foreach \x in {4,5} {
        \foreach \y in {5} {
          \node at (\x,\y) (x\x\y) [var] {};
        }
      }
      \foreach \x in {4} {
        \foreach \y in {4} {
          \node at (\x,\y) (x\x\y) [var] {};
        }
      }
      \foreach \x in {9} {
        \foreach \y in {0,1,2,3,4,5,6,7,8} {
          \node at (\x,\y) (x\x\y) [obs] {};
        }
      }
      \foreach \x in {1,2,3,4,5,6,7,8} {
        \foreach \y in {8} {
          \node at (\x,\y) (x\x\y) [obs] {};
        }
      }
      \foreach \x in {1} {
        \foreach \y in {1,2,3,4,5,6,7} {
          \node at (\x,\y) (x\x\y) [obs] {};
        }
      }
      \foreach \x in {2,3,4,5,6,7} {
        \foreach \y in {1} {
          \node at (\x,\y) (x\x\y) [obs] {};
        }
      }
      \foreach \x in {7} {
        \foreach \y in {2,3,4,5,6} {
          \node at (\x,\y) (x\x\y) [obs] {};
        }
      }
      \foreach \x in {3,4,5,6} {
        \foreach \y in {6} {
          \node at (\x,\y) (x\x\y) [obs] {};
        }
      }
      \foreach \x in {3} {
        \foreach \y in {3,4,5} {
          \node at (\x,\y) (x\x\y) [obs] {};
        }
      }
      \foreach \x in {4,5} {
        \foreach \y in {3} {
          \node at (\x,\y) (x\x\y) [obs] {};
        }
      }
      \foreach \x in {5} {
        \foreach \y in {4} {
          \node at (\x,\y) (x\x\y) [obs] {};
        }
      }
      \foreach \x in {0,1,2,3,4,5,6,7,8} {
        \pgfmathtruncatemacro\xend{\x+1}
        \foreach \y in {0,1,2,3,4,5,6,7,8,9} {
          \draw[edge] (x\x\y) -- (x\xend\y) {};
        }
      }
      \foreach \x in {0,1,2,3,4,5,6,7,8,9} {
        \foreach \y in {0,1,2,3,4,5,6,7,8} {
          \pgfmathtruncatemacro\yend{\y+1}
          \draw[edge] (x\x\y) -- (x\x\yend) {};
        }
      }
    \end{scope}
  \end{tikzpicture}
    \label{fig:blocking}
    }
  \end{subfigure}
  \caption{{\em Left:} Ordering of factors in the \pgas algorithm. At each iteration all the
    factors connecting the added node and the previous nodes are included in the target
    distribution. {\em Right:} The two block structures used in the tree-sampler and \pgas with partial blocking. Nodes are added from the edge spiralling in.}\label{fig:blockorder}
\end{figure}
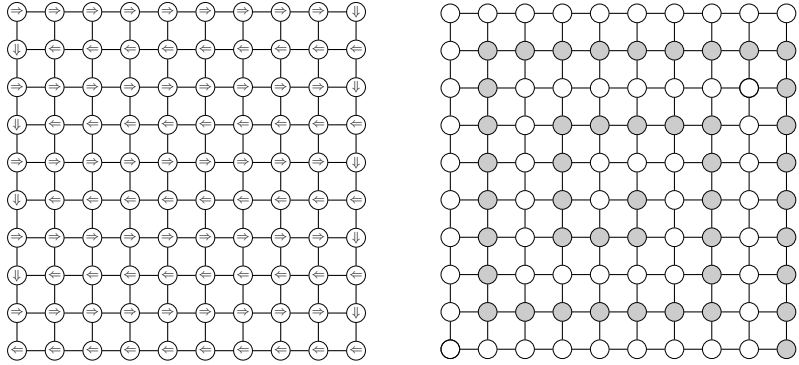

In Figure~\ref{fig:gmrfacf} we can see the empirical autocorrelation funtions (ACFs) centered around the true posterior mean for variable $x_{82}$ (selected randomly from $\X_\V$).
Similar results hold for all the variables of the model.  Due to the strong interaction between the
latent variables, the samples generated by the standard Gibbs sampler are strongly correlated.
Tree-sampling and \pgas with partial blocking show nearly identical gains compared to Gibbs. This is
interesting, since it suggest that simulating from the \smc-based \pgas kernel can be almost as
efficient as exact simulation, even using a moderate number of particles.  We emphasize that the
\pgas kernels leave their respective target distributions invariant, \ie the limiting distributions
is the same for all \mcmc schemes.  The fully blocked \pgas algorithm achieves the best ACF,
dropping off to zero considerably faster than for the other methods. This is not surprising since
this sampler simulates all the latent variables jointly which reduces the autocorrelation, in
particular when the latent variables are strongly dependent.


\begin{figure}[h!]
  \centering
  \begin{subfigure}{0.4\textwidth}
    \includegraphics[width=\textwidth]{./10x10acf}
    \caption{The empirical ACF for Gibbs sampler, \pgas, \pgas with partial blocking and tree
      sampler. Based on $100\thinspace000$ data points with $10\%$ burnin.}\label{fig:gmrfacf}
  \end{subfigure}%
  ~
  \begin{subfigure}{0.4\textwidth}
  \includegraphics[width=\textwidth]{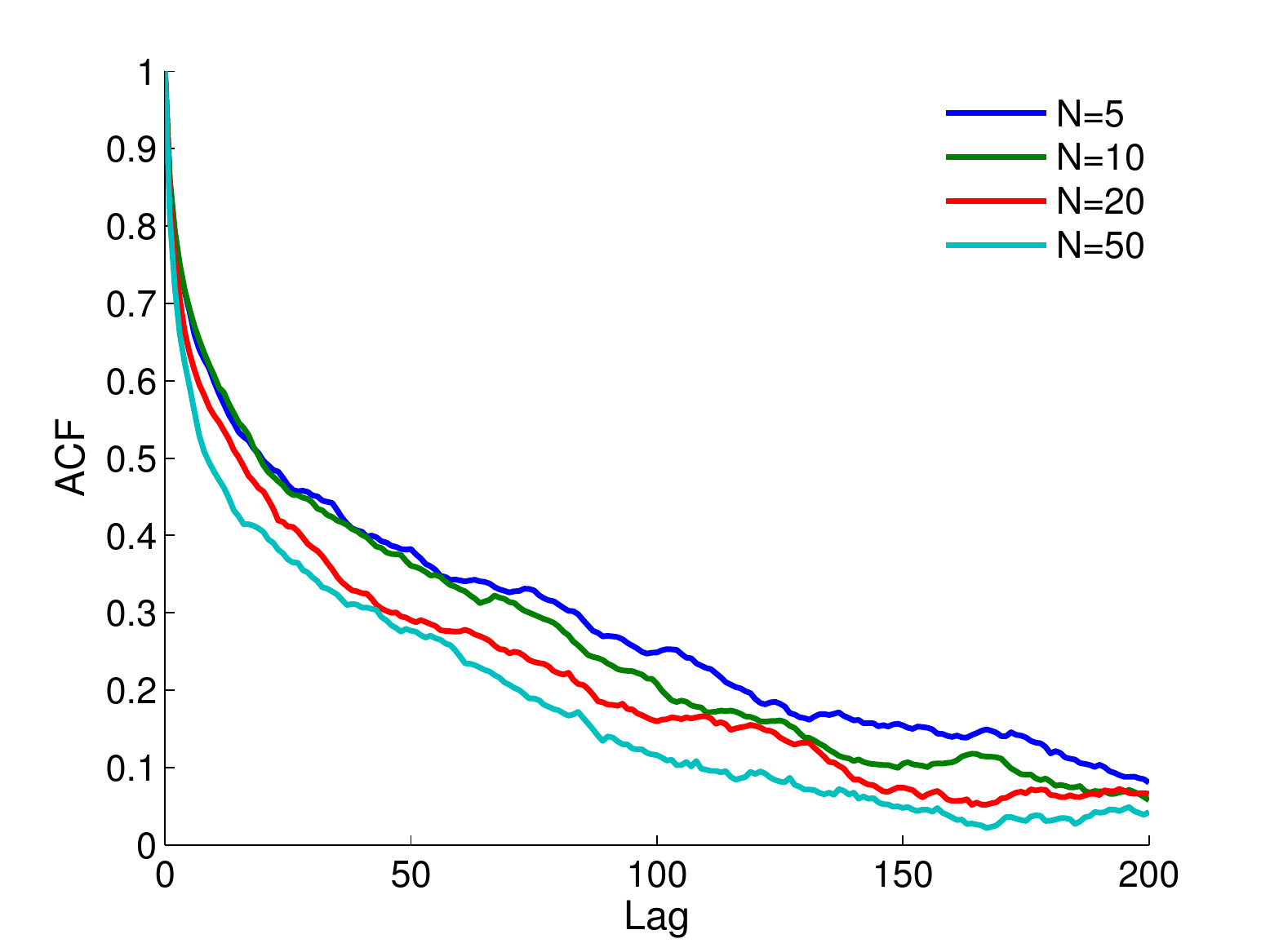}
    \caption{The empirical ACF for \pgas with partial blocking for $N=5,10,20,50$. Based on $10\thinspace000$ data points with $10\%$ burnin.}\label{fig:gmrfacf}
  \end{subfigure}%
\end{figure}

However, this improvement in autocorrelation comes at a cost. For the fully blocked \pgas kernel, the maximal cardinality of the
set $\A_k$ (see \eqref{eq:pgas:Ak-def}) is 10 (one full row of variables). For the partially blocked \pgas kernel, on the other hand, $|A_k| \equiv 1$ since the
variables in each block form a chain. This implies that the fully blocked \pgas sampler is an order of magnitude more computationally involved than
the partially blocked \pgas sampler. This trade-off between autocorrelation and computational efficiency has to be taken into account when
deciding which algorithm that is most suitable for any given problem.

\subsection{Proof of unbiasedness}
Recall that we use the convention $\xi_k = \X_{\ind_k \setminus \Ind_{k-1}}$. Define recursively the functions $f_k(\Xk) \equiv 1$ and,
\begin{align}
  \label{eq:f-def}
  f_{\ell}(\Xk[\ell]) = \frac{ \int f_{\ell+1}(\Xk[\ell+1]) \gamma_\ell(\Xk[\ell+1]) \rmd \xi_{\ell+1} }{\gamma_{\ell}(\Xk[\ell])}
\end{align}
for $\ell = k-1,\, k-2,\, \dots,\, 1$. Let
\begin{align*}
  Q_\ell = \left( \frac{1}{N} \sum_{i=1}^N  w_\ell^i f_\ell(\Xk[\ell]^i) \right) \left\{ \prod_{m = 1}^{\ell-1}  \frac{1}{N} \sum_{i=1}^N \nu_m^i w_m^i \right\},
\end{align*}
for $\ell \in \set{1}{k}$. Note that, by construction, $Q_k = \widehat Z_k^N$. Let $\filt_\ell$ be the filtration generated
by the particles simulated up to iteration $\ell$:
\begin{align*}
  \filt_\ell \eqdef \sigma( \{\Xk[m]^i, w_m^i\}_{i=1}^N, m = \range{1}{\ell}).
\end{align*}

\begin{lem}
  The sequence $\{ Q_\ell, \ell = \range{1}{k}\}$ is an $\filt_\ell$-martingale.
\end{lem}
\begin{proof}
  Consider,
  \begin{align*}
    \EE\left[ Q_\ell \mid \filt_{\ell-1} \right] = \EE\left[ w_\ell^1 f_\ell(\Xk[\ell]^1) \mid \filt_{\ell-1} \right]\left\{ \prod_{m = 1}^{\ell-1}  \frac{1}{N} \sum_{i=1}^{N} \nu_m^i w_m^i \right\}.
  \end{align*}
  Using the definition of the weight function (see the main document) we have,
  \begin{align*}
    &\EE\left[ W_\ell(\Xk[\ell]^1) f_\ell(\Xk[\ell]^1) \mid \filt_{\ell-1} \right] \\
    &\qquad = \sum_{i=1}^N \int W_\ell( \{ \Xk[\ell-1]^i \cup \xi_\ell \} ) f_\ell( \{ \Xk[\ell-1]^i \cup \xi_\ell \} )
    \frac{\nu_{\ell-1}^i w_{\ell-1}^i }{\sum_l \nu_{\ell-1}^l w_{\ell-1}^l } r_\ell(\xi_\ell | \Xk[\ell-1]^i) \rmd \xi_\ell \\
    &\qquad = \frac{1}{\sum_l \nu_{\ell-1}^l w_{\ell-1}^l }  \sum_{i=1}^N \left( \nu_{\ell-1}^i w_{\ell-1}^i
    \frac{ \int \gamma_\ell( \{ \Xk[\ell-1]^i \cup \xi_\ell \} ) f_\ell( \{ \Xk[\ell-1]^i \cup \xi_\ell \} ) \rmd \xi_\ell }{ \gamma_{\ell-1}(\Xk[\ell-1]^i)  \nu_{\ell-1}(\Xk[\ell-1]^i) } \right) \\
    &\qquad = \frac{1}{\sum_l \nu_{\ell-1}^l w_{\ell-1}^l }  \sum_{i=1}^N \left( w_{\ell-1}^i f_{\ell-1}(\Xk[\ell-1]^i) \right).
  \end{align*}
  Hence, we get
  \begin{align*}
    \EE\left[ Q_\ell \mid \filt_{\ell-1} \right] = 
     \sum_i \left( w_{\ell-1}^i f_{\ell-1}(\Xk[\ell-1]^i) \right)   \frac{1}{N}   \left\{ \prod_{m = 1}^{\ell-2}  \frac{1}{N} \sum_i \nu_m^i w_m^i \right\} = Q_{\ell-1}.
  \end{align*}
\end{proof}
It follows that
\begin{align*}
  \EE[ \widehat Z_k^N ] = \EE[Q_k] = \EE[Q_1] = \int W_1(\Xk[1]) f_1(\Xk[1]) r_1(\Xk[1]) \rmd \Xk[1] = \int \gamma_1(\Xk[1]) f_1(\Xk[1]) \rmd \Xk[1].
\end{align*}
However, from the definition in \eqref{eq:f-def} we have that
\begin{align*}
  \int \gamma_1(\Xk[1]) f_1(\Xk[1]) \rmd \Xk[1] &= \iint \gamma_2(\Xk[2]) f_2(\Xk[2]) \rmd \Xk[2] \\
  &= \cdots  = \int\tcdots\int \gamma_k(\Xk) f_{k}(\Xk) \rmd \Xk = Z_k.
\end{align*}
\hfill\qedsymbol

\subsection{Ancestor sampling}
To implement the \pgas sampling procedure, it remains to detail the ancestor sampling step.
At each iteration $k \geq 2$, this step amount to generating a value for the ancestor index $a_k^N$ corresponding to the reference particle.
Implicitly, this assigns an artificial history for the ``remaining'' part of the reference
particle $\X_{\Ipgas \setminus \Ind_{k-1}}^\prime$, by selecting one of the particles $\{ \Xk[k-1]^i \}_{i=1}^N$ as its ancestor. This
results in a complete assignment for the collection of latent variables of the \pgm
\begin{align}
  \label{eq:pgas:concatenated-particle}
  \XtV_{\Ipgas}^i \eqdef \{ \Xk[k-1]^{i} \cup \X_{\Ipgas \setminus \Ind_{k-1}}^\prime \} \in \setXJ.
\end{align}
As shown by \citet{LindstenJS:2014}, the probability distribution from which $a_k^N$ should be sampled in order to ensure reversibility of the \pgas kernel \wrt $\bar\gamma_K$ is given by,
\begin{align}
  \label{eq:pgas:weights}
\Prb(a_k^N = i) \propto w_{k-1}^i
\frac{\displaystyle \gammaK \left( \XtV_{\Ipgas}^i \right)}{\displaystyle \gamma_{k-1} (\Xk[k-1]^i)}.
\end{align}
This expression can be understood as an application of Bayes' theorem, where $w_{k-1}^i$ is the prior probability of particle $\Xk[k-1]^i$
and the ratio between the target densities is the unnormalized likelihood of 
$\X_{\Ipgas \setminus \Ind_{k-1}}^\prime$
conditionally on $\Xk[k-1]^i$.

To derive an explicit expression for the ancestor sampling probabilities in our setting, note first that,
\begin{align}
  \frac{\gammaK \left( \X_{\Ipgas} \right)}{\gamma_{k-1} \left( \Xk[k-1] \right)} =
  \frac{ \prod_{j = k}^K \psi_j \left(\xk[j] \right) }{ q_{k-1}(\Xk[k-1]) }.
\end{align}
Now, let $\A_k$ be the index set of factors $\psi_j$, $j \geq k$ for which any of the variables $\Xk[k-1]$ is in the domain of $\psi_j$; formally
\begin{align}
  \label{eq:pgas:Ak-def}
  \A_k \eqdef \{  j : k \leq j \leq K, \Ind_{k-1} \cap \ind_j \neq \emptyset \}.
\end{align}
It follows that any factor $\psi_j$ for which $j \not\in \A_k $ is independent of $\Xk[k-1]$. Consequently, we can write \eqref{eq:pgas:weights}
as
\begin{align}
  \label{eq:pgas:weights-2}
  \Prb(a_k^N = i) \propto w_{k-1}^i
  \frac{ \prod_{j \in \A_k} \psi_j \left( \XtV_{\ind_j}^i \right) }{ q_{k-1}(\Xk[k-1]^i) },
\end{align}
where $\XtV_{\ind_j}^i$ is a subset of the variables in \eqref{eq:pgas:concatenated-particle}. 
In fact, the index set $\A_k$ corresponds exactly to the factors $\psi_j$ that depend, explicitly, both on the particle $\Xk[k-1]^i$
and on the reference particle $\X_{\Ipgas \setminus \Ind_{k-1}}^\prime$ (through some of their respective components). Indeed, it is only these factors
that hold any information about the likelihood of $\X_{\Ipgas \setminus \Ind_{k-1}}^\prime$ given $\Xk[k-1]^i$.

The expression \eqref{eq:pgas:weights-2} is interesting, since it shows that the computational complexity of the ancestor
sampling step will depend on the cardinality of the set $\A_k$.
While this will depend both on the structure of the graph and on the ordering of the factors, it will for many models of interest be of a lower order than the
cardinality of $\V$.


\renewcommand{\baselinestretch}{0.95}

\def\bibfont{\small}
\bibliographystyle{unsrtnat}
\bibliography{smcforpgm}

\end{document}